\documentclass{article}

\usepackage{PRIMEarxiv}
\usepackage{blindtext}
\usepackage{natbib}
\usepackage[T1]{fontenc}
\usepackage[utf8]{inputenc}
\usepackage{microtype}
\usepackage{upgreek}
\usepackage{amsmath,amssymb,amsthm,mathtools,bm,mathrsfs}
\usepackage{bbm}
\usepackage{dsfont}
\usepackage{url}
\usepackage{algorithm}
\usepackage{algpseudocode}
\usepackage{tikz}
\usetikzlibrary{arrows.meta,positioning,calc,backgrounds, decorations.pathreplacing}

\usepackage{tabularx}

\usepackage{booktabs}
\usepackage{graphicx}
\graphicspath{{media/}}

\usepackage{xcolor}

\usepackage{lastpage}
\usepackage{amssymb}
\usepackage{pifont}
\newcommand{\cmark}{\ding{51}}%
\newcommand{\xmark}{\ding{55}}%
\newtheorem{thm}{Theorem}
\newtheorem{lem}{Lemma}

\newtheorem{rmrk}{Remark}

\newtheorem{cor}{Corollary}
\newtheorem{assumption}{Assumption}

\newcommand{\R}{\mathbb{R}}

\newcommand{\E}{\mathbb{E}}
\renewcommand{\Pr}{\mathbb{P}}

\newcommand{\Ec}{\mathcal{E}}
\newcommand{\Var}{\mathrm{Var}}

\newcommand{\V}{\mathcal{V}}
\newcommand{\A}{\mathcal{A}}

\newcommand{\As}{\mathscr{A}}
\newcommand{\F}{\mathcal{F}}
\newcommand{\G}{\mathcal{G}}

\newcommand{\Unif}{\mathrm{Unif}}
\newcommand{\one}{\mathbf 1}

\newcommand{\Jset}{\mathcal{J}}
\newcommand{\Oc}{\mathcal O}

\newcommand{\algmark}[1]{%
  \tikz[remember picture,overlay]\coordinate (#1);%
}

\newcommand{\algrbrace}[3]{%
  \tikz[remember picture,overlay]
    \draw[decorate,decoration={brace,amplitude=5pt},color=blue]
      ($(#1)+(22em,0.5ex)$) -- ($(#2)+(22em,-0.5ex)$)
      node[midway,xshift=3.2em,align=left] {\scriptsize #3};
}

\title{Virtual Dummies: Enabling Scalable FDR-Controlled Variable Selection via Sequential Sampling of Null Features}
\author{Taulant Koka, Jasin Machkour, Daniel P. Palomar, Michael Muma}

\begin{document}
\maketitle

\begin{abstract}
High-dimensional variable selection, particularly in genomics, requires error-controlling procedures that scale to millions of predictors. The Terminating-Random Experiments (T-Rex) selector achieves false discovery rate (FDR) control by aggregating results of early terminated random experiments, each combining original predictors with i.i.d. synthetic null variables (dummies). At biobank scales, however, explicit dummy augmentation requires terabytes of memory. We demonstrate that this bottleneck is not fundamental. Formalizing the information flow of forward selection through a filtration, we show that compatible selectors interact with unselected dummies solely through projections onto an adaptively evolving low-dimensional subspace. For rotationally invariant dummy distributions, we derive an adaptive stick-breaking construction sampling these projections from their exact conditional distribution given the selection history, thereby eliminating dummy matrix materialization. We prove a pathwise universality theorem: under mild delocalization conditions, selection paths driven by generic standardized i.i.d. dummies converge to the same Gaussian limit. We instantiate the theory through Virtual Dummy LARS (VD-LARS), reducing memory and runtime by several orders of magnitude while preserving the exact selection law and FDR guarantees of the T-Rex selector. Experiments on realistic genome-wide association study data confirm that VD-T-Rex controls FDR and achieves power at scales where all competing methods either fail or time out.
\end{abstract}

\section{Introduction}
\label{sec:intro}

The first eukaryotic genome sequence, that of \emph{Saccharomyces cerevisiae}\footnote{Commonly known as baker's yeast or brewer's yeast.}, was completed in 1996 by \citet{Goffeau1996}, marking a turning point in molecular biology. The same year, \citet{Tibshirani1996} introduced the Lasso, laying the foundation for modern high-dimensional variable selection. At the time these advances may have seemed unrelated; yet, three decades later, they converge conceptually in large-scale genomic studies such as genome-wide association studies (GWAS) \citep{Uffelmann2021}, where the statistical task is to identify a small active set among millions of candidate predictors.

The standard mathematical setting is the high-dimensional linear model
\[
\bm y = \bm X \bm\beta + \bm\upepsilon,
\qquad
\bm\upepsilon \sim \mathcal N(\bm 0,\sigma^2\bm I_n),
\]
where $n, p$ are the numbers of observations and predictors, respectively; $\bm y\in\R^n$ is the response, $\bm X\in\R^{n\times p}$ is the design matrix, $\bm\beta\in\R^p$ is the unknown sparse coefficient vector, and the goal is to recover the active set
\(
\A = \{ j : \beta_j \neq 0 \}
\).

The Lasso promotes sparsity through an $\ell_1$ penalty, and its geometry leads naturally to forward selection (FS) procedures that incrementally add predictors. The \emph{Least Angle Regression (LARS)} algorithm \citep{Efron2004} is a canonical example: it advances along the equiangular direction of active predictors and, with a slight modification, traces the Lasso path exactly. However, classical FS methods, including LARS, do not provide false discovery rate (FDR) control, a critical limitation in large-scale problems where $p\gg n$. Indeed, \cite{Su2017} showed that false discoveries are interspersed among true ones early on the Lasso path, motivating procedures that explicitly calibrate the rate of null variable inclusion.

A recent strategy for controlling the FDR in high-dimensional settings is to add synthetic \emph{null} variables that mimic the behavior of the true features under the null. The \emph{model--X knockoff} framework \citep{Cands2018, Ren2021, ren2024} follows this idea by generating knockoff copies of the predictors: variables that match the joint distribution of $\bm{X}$ yet are independent of the response. Comparing each feature to its knockoff yields exact FDR control, but constructing such knockoffs requires modeling the full joint distribution of $\bm{X}$, which is often computationally prohibitive when $p$ is large.

The \emph{Terminating Random Experiments (T--Rex)} selector \citep{Machkour2025,Machkour2025_Dependent} offers a scalable FDR-controlling alternative that avoids modeling the covariance matrix of $\bm{X}$. The T--Rex selector conducts multiple randomized experiments in which $\bm{X}$ is augmented with $L$ synthetic null predictors, so-called \emph{dummy variables} or \emph{dummies}. In each of the independent random experiments, the dummies compete against the real predictors in a forward-selection process. These dummies are drawn independently from some distribution~$\Psi$:
\[
  \bm{D} = (\bm{d}_1,\dots,\bm{d}_L), 
  \qquad 
  \bm{d}_\ell \stackrel{\text{i.i.d.}}{\sim} \Psi.
\]
In the original T--Rex formulation, the dummy design follows an $n$-fold product distribution $\Psi=\mu^{\otimes n}$ for some univariate base law $\mu$, where \cite{Machkour2025,Machkour2025_Dependent} adopt independent standard normal entries ($\mu=\mathcal N(0,1)$) as a convenient default choice. 
Despite its statistical guarantees, T--Rex faces a severe computational bottleneck. Each random experiment requires (i) explicitly materializing an $n\times L$ dummy block, typically with $L\ge p$, and (ii) repeatedly computing correlations against the residual across all $p+L$ predictors along the forward path. At biobank scales (e.g., $n=5{\times}10^5$, $p=10^6$, $L\ge p$), storing the dummy block alone would require over $4$~TB in \texttt{float64}, well beyond the capacity of typical machines. While out-of-core techniques such as memory mapping can mitigate RAM constraints \citep{Scheidt2023}, they do not remove the fundamental cost of explicitly handling and repeatedly processing an $n\times L$ dummy block.

We show that this bottleneck is not fundamental. The key observation is that many forward selection procedures, including LARS, orthogonal matching pursuit (OMP) \citep{pati1993orthogonal}, stepwise regression \citep{Hocking1976}, and others, never require access to the full dummy matrix itself. Instead of using the full coordinates of a dummy $\bm d_\ell$, the selection rule only requires its correlation with the current residual. More generally, any decision to add a variable depends solely on the dummy's projections onto the subspace formed by the response and the variables already selected. Thus, forward selection never uses the $n$ coordinates of a dummy directly, but only a growing collection of low-dimensional projections along the selection path.

In general, however, these projections cannot be generated independently of the unrevealed coordinates, since they are constrained by their joint distribution. A central challenge is therefore to characterize the conditional distribution of future projections given the revealed past.

This difficulty admits a precise resolution for suitably standardized \emph{rotationally invariant} dummy distributions.\footnote{A distribution
$\Psi$ on $\R^n$ is rotationally invariant if $\bm Q\bm d \stackrel d= \bm d$ for all orthogonal $\bm Q$.} For Gaussian and uniform spherical base laws, rotational invariance implies that, after centering and normalization, the conditional distribution of the unrevealed components of a dummy depends only on the already revealed subspace, not on the specific data-dependent basis used to represent it. One of the main tasks of this paper is to show that this invariance is strong enough to permit sequential sampling of dummy projections along an adaptively evolving orthonormal basis from their exact conditional distribution.

At first glance, restricting attention to Gaussian or uniform spherical dummy distributions may appear limiting. For large $n$, however, this restriction is often immaterial for the early part of a forward-selection run. Beyond our finite-sample exact construction under rotational invariance, we prove a pathwise universality result in a fixed-$(K,L,p)$ regime: for any fixed numbers of steps $K$, predictors $p$, and dummies $L$, and under mild moment and delocalization conditions, the $K$-step forward-selection path generated by standardized i.i.d. dummies converges in distribution, as $n\to\infty$, to the same limiting selection process as the path generated by Gaussian dummies. Thus, Gaussian (and spherical) dummies serve as canonical proxies for the initial trajectory, while our finite-sample results provide exact guarantees under rotational invariance for any fixed finite $L$.

As a consequence, under rotationally invariant dummy laws, the explicit $n\times L$ dummy matrix never needs to be materialized. Instead, dummy variables can be represented implicitly through the collection of projections along the forward selection path, whose dimension grows only linearly with the number of selected directions. In large-scale settings, this reduces the \emph{dummy-related} memory footprint from terabytes to a few hundred megabytes, effectively removing the primary computational bottleneck of dummy-augmented methods such as the T--Rex selector.

Beyond the specific context of the T--Rex, the virtual dummy perspective offers a general template for accelerating any randomized variable selection method relying on rotationally invariant augmentation. By removing the computational barriers inherent to these high-dimensional synthetic variable schemes, the proposed virtual-dummy construction paves the way for multivariate FDR-controlled variable selection at biobank scale, offering a path toward reproducible discoveries of truly relevant genetic variants associated with diseases.

\subsection{Contributions}
This paper develops a general construction, proposes an associated methodology, and establishes theoretical guarantees for integrating virtual-dummies into forward selection. Our main contributions are:
\begin{enumerate}
  \item \textbf{Sequential sampling.} We formalize the information revealed by forward selection via a filtration, characterize compatible selection rules, and, under rotationally invariant dummy laws, derive an adaptive stick-breaking representation that enables valid sequential sampling of spherical dummies under data-dependent basis updates.
  \item \textbf{Exactness and universality.} We show that dummy-augmented and virtual-dummy forward selectors have identical conditional laws of the entire randomized forward selection path, so existing FDR guarantees (in particular for T--Rex) carry over without modification. We further prove a pathwise universality result: the finite-dimensional dummy projections that determine the first $K$ forward-selection steps (for any fixed $K$, and in particular for early-terminated T--Rex paths) converge to the Gaussian limit with respect to $n$, making the choice of standardized i.i.d. dummy law asymptotically irrelevant for the first $K$ steps.
    \item \textbf{Algorithm and empirical validation.} We instantiate the construction for LARS, yielding the \emph{Virtual Dummy LARS  (VD--LARS)}, which avoids forming the $n\times L$ dummy matrix and reduces dummy-related memory and runtime by more than two orders of magnitude. Simulations confirm that VD--LARS preserves the FDR control and power of the T--Rex selector while scaling to dimensions that are multiple orders of magnitude beyond the failure point of explicit augmentation. We further benchmark our method on large-scale GWAS data simulated with realistic linkage disequilibrium patterns. An open-source \texttt{C++} implementation of VD--LARS, VD--OMP, VD--AFS (including logistic regression), and VD--T--Rex, together with code to reproduce all experiments in this paper, is available at \url{https://github.com/taulantkoka/virtual-dummies}.
\end{enumerate}

\subsection{Related Work}
\label{ssec:related}

The virtual-dummy framework provides a computationally efficient realization of uninformative negative controls \citep{Miller1984,Wu2007} by replacing explicit design augmentation with sequentially generated projections. While the knockoff framework \citep{Barber2015,Cands2018,Romano2019,pmlr-v48-daia16,Li2021} achieves FDR control by constructing surrogates that mimic the dependence structure of $\bm X$ (see also \citealt{Ke2024} for a power analysis of knockoff design choices), and mirrored statistics \citep{Xing2021,Chen2023} exploit the sign symmetry of null test statistics, our approach instead exploits rotational invariance of the dummy distribution to generate projections directly in the adaptively revealed basis, while preserving the augmented law. 

Knockoff generation itself can be cast as a sequential sampling problem: Sequential Conditional Independent Pairs (SCIP) \citep{Cands2018} and Metropolized knockoff sampling \citep{Bates2020} generate knockoff copies of $\bm X$ one variable at a time by exploiting conditional independence structure, while FANOK \citep{Askari2021} accelerates the knockoff construction by exploiting factor model structure in the semidefinite program (SDP). Both lines of work aim to produce the \emph{entire} knockoff matrix more efficiently. If the entries of $\bm X$ are i.i.d.\ draws from some univariate distribution, the pairwise exchangeability holds trivially for independent copies drawn from the same distribution, and knockoffs reduce to i.i.d.\ dummies. For general $\bm\Sigma$, however, knockoffs must model the full dependence structure of $\bm X$, whereas dummies remain independent by design. By contrast, the virtual-dummy construction proposed here does not generate full $n$-dimensional vectors at all; instead, it samples only the low-dimensional projections that the forward-selection path actually requires, avoiding materialization of the dummy matrix entirely.

In contrast to subsampling-based stability methods such as the Bolasso \citep{Bach2008} or stability selection \citep{Meinshausen2010,Shah2012}, T--Rex \citep{Machkour2025,Machkour2025_Dependent} uses dummy-based calibration to stabilize forward-selection paths. Beyond its original regression setting, T--Rex has been applied to graphical model estimation \citep{Koka2024}, sparse PCA \citep{Machkour2024}, grouped selection \citep{Machkour2023_InformedElasticNet}, and portfolio tracking \citep{Machkour2025_Portfolio}. A dummy-free variant of the T--Rex selector, based on tracking dummy importance
along the forward-selection path, is proposed as a future research direction in \citet{Machkour2024phd}. In contrast, the present work retains the dummy-based construction but replaces explicit dummy variables by an implicit representation, yielding an exactly distribution-preserving procedure with substantially reduced computational cost. Because the virtual-dummy construction preserves the full law of the augmented path (Theorem~\ref{thm:vd-fs-dist}), existing benchmarks against knockoff-based procedures (see, \cite{Machkour2025,Machkour2025_Dependent,Koka2024}) remain statistically unchanged, allowing us to focus on the computational gains over explicit augmentation.

\subsection*{Notation}
Boldface letters denote vectors and matrices (e.g., $\bm x$, $\bm X$); calligraphic symbols denote sets, $\sigma$-algebras and filtrations (e.g., $\F_k$), and blackboard symbols denote probability spaces and expectations. For a set of random variables $\{\zeta_\ell\}_{\ell =1}^L=\{\zeta_1, \dots, \zeta_L\}$ and an event $E_\ell$, we write
$\{\zeta_\ell\}_{E_\ell}$ for all $\zeta_\ell$ with $E_\ell$ true.  
If $\tau_\ell$ is the selection time of dummy $\ell$, then $\{\zeta_\ell\}_{\tau_\ell>k}$ collects the quantities associated with dummies remaining unselected at step $k$; this appears, for example, in the conditional update $\{\alpha_{k+1,\ell}\}_{\tau_\ell>k}$.
Real variables use indices $j\in\{1,\dots,p\}$ and dummies use
$\ell\in\{1,\dots,L\}$.   A dummy with local index $\ell$ appears in the augmented design at global index $p+\ell$. Inner products are written $\langle\cdot,\cdot\rangle$, and $\sigma(\cdot)$ denotes generated $\sigma$-algebras. All random objects are defined on the probability space $(\Omega,\mathcal F,\mathbb P)$ introduced later. For any subspace $\V\subseteq\R^n$, let  $\bm P_{\,\V}$ denote the orthogonal projector onto $\V$. We denote by $\V^\perp$ the orthogonal complement of $\V$ in the ambient space. The standard orthogonal group is denoted by
\(
\mathrm O(n) := \{\bm Q \in \R^{n \times n} : \bm Q^\top \bm Q = \bm I_n \}.
\) For a subspace $\V\subseteq\R^n$, we write
\(
\mathrm O(\V) := \{ \bm Q \in \mathrm O(n) : \bm Q\V=\V \}
\) for the subgroup that preserves $\V$ (setwise), and
\(
\mathrm O(\V^\perp \mid \V)
:= \{ \bm Q \in \mathrm O(n) : \bm Q\bm v = \bm v \ \ \forall\,\bm v \in \V \}
\) for the pointwise stabilizer of $\V$ (equivalently, orthogonal transformations acting freely on $\V^\perp$ while fixing $\V$).
\section{Preliminaries}
\label{sec:setup}
\vspace{1ex}

We formalize the probabilistic and geometric setting underlying the
\emph{virtual-dummy forward-selection} framework. At each step, a
forward-selection procedure interacts with each candidate variable only
through its projection onto the subspace spanned by the response and the
previously selected directions. The distribution of these projections must
therefore be interpreted conditionally on the information revealed so far.
To capture this, we construct a filtration that records all revealed dummy
projections and realized dummy vectors, and track how an orthonormal basis
for the revealed subspace $\V_k$ evolves with this filtration.
This framework will later allow us to prove that augmented and virtual-dummy
procedures share the same transition kernels.

\medskip
\noindent
We consider a variable selection problem with response vector
\(\bm y \in \R^n\) and design matrix
\(\bm X = (\bm x_1,\dots,\bm x_p) \in \R^{n\times p}\),
where both the response and the predictors are centered:
\(\one^\top \bm y = 0\) and \(\one^\top \bm x_j = 0\) for all \(j\).
All vectors therefore lie in the centered subspace
\[
H := \{\,\bm x \in \R^n : \one^\top \bm x = 0\,\},
\qquad
m := \dim(H) = n-1 \ge 2.
\]

In our setting, a forward-selection procedure acts on the augmented design
\[
(\bm X\ \bm D) = (\bm x_1,\dots,\bm x_p,\bm d_1,\dots,\bm d_L)
\in H^{p+L},
\]
sequentially choosing one column at each step according to a data-dependent rule that may depend on \((\bm X,\bm y)\) and on the progressively revealed dummy projections (see Table~\ref{tab:scope_compatibility} for a comprehensive list).
By orthogonalizing each selected column against the previous directions, we obtain a growing orthonormal basis
\[
\Ec_m := \{\bm e_1,\bm e_2,\dots,\bm e_m\} \subset H,
\qquad
\langle\bm e_i,\bm e_j\rangle = \delta_{ij},\quad \|\bm e_i\|_2=1,
\]
for the revealed subspace
$\V_k := \mathrm{span}(\bm e_1,\dots,\bm e_k)$ after step~$k$.
The dummy variables are drawn i.i.d.\ from a rotationally invariant base law~$\Psi(H)$:
\[
\bm d_1,\dots,\bm d_L \stackrel{\text{i.i.d.}}{\sim} \Psi(H).
\]
For the conditional laws induced by $\Psi(H)$, we write $\Psi_k$ for the one--dimensional conditional law of $\langle \bm d, \bm e_k\rangle$, $\Psi_k^\parallel$ for the $k$--dimensional conditional law on $\V_k$, and $\Psi_k^\perp$ for the $(m-k)$--dimensional conditional law on $H\cap\V_k^\perp$, interpreted as kernels given the revealed subspace. All randomness arises from the dummy ensemble; $(\bm X,\bm y)$ are treated as fixed throughout. Hence, we work on the canonical probability space
\begin{equation}
\label{eq:probability_space}
(\Omega,\F,\mathbb P)
~=~
\bigl(H^L,\,\mathcal B(H^L),\,\Psi(H)^{\otimes L}\bigr),
\end{equation}
where $\bm d_\ell(\omega)=\omega_\ell$ and $\mathcal B(H^L)$ is the Borel $\sigma$-algebra over $H^L$.

Although augmented forward selection is described as operating on the full dummy matrix $\bm D$, the algorithm never needs to access $\bm D$ explicitly. After $k$ steps, the procedure has revealed the directions $\bm e_1,\dots,\bm e_k$ and, for each dummy $\bm d_\ell$, only the projections $\langle \bm d_\ell,\bm e_i\rangle$ for $i\le k$. For dummies that have not been selected, the remaining component in $H\cap \V_k^\perp$ is still completely unseen and plays no role in the next decision. Only when a dummy is selected do we need to ``complete'' it by sampling its orthogonal residual.

 In the following section, we formalize this reveal-on-demand mechanism via a filtration $(\F_k)$ whose evolution follows
\[
\F_k \;\to\; \F_k^+ \;\to\; \F_{k+1},
\]
where $\F_k$ is the information used to choose the next variable, $\F_k^+$ additionally reveals any dummy selected at step $k$ (thereby fixing $\bm e_{k+1}$), and $\F_{k+1}$ then reveals the next round of dummy projections along $\bm e_{k+1}$. The definitions below make this information pattern precise and will later allow us to compare augmented and virtual--dummy procedures via their conditional transition laws.

\subsection{Filtration of the Forward Selection}
\label{ssec:filtration}
The progressive revelation of dummy information to the forward-selection
algorithm is encoded by a filtration $(\F_k)_{k=0}^m$ on $(\Omega,\F)$, i.e. an increasing family of sub-$\sigma$--algebras with $\F_k\subseteq\F$, which we now construct explicitly.
An overview of its time evolution is given in
Figure~\ref{fig:filt_timeline}. The cumulative information available \emph{immediately before} selecting the next
variable is represented by $\F_k$.
We initialize
\[
\F_0 := \sigma(\bm X,\bm y),
\qquad
\bm e_1 := \frac{\bm y}{\|\bm y\|_2}\in H.
\]
Thus, \(\F_0\) contains no information about any dummy projection yet\footnote{For notational convenience, we write $\sigma(\bm X,\bm y)$ inside filtrations even though $(\bm X,\bm y)$ are deterministic on this space, so that $\sigma(\bm X,\bm y)=\{\varnothing,\Omega\}$.}.
Before advancing to step $k=1$, the first projections of all dummies $\bm d_\ell,\:\ell=1,\dots,L$ onto \(\bm e_1\) are drawn from the one--dimensional law induced by $\Psi(H)$ and revealed:
\[
\upalpha_{1\ell} := \langle \bm d_\ell, \bm e_1\rangle,
\qquad
(\upalpha_{11},\dots,\upalpha_{1L}) \stackrel{\text{i.i.d.}}{\sim} \Psi_1,
\]
Accordingly,
\[
\F_1 := \sigma\!\bigl(\F_0,\,\{\upalpha_{1\ell}\}_{\ell=1}^L\bigr).
\]
At this stage, the procedure has seen one coordinate of each dummy, but none of the remaining coordinates in the orthogonal complement.
Let the global index set of candidate variables \(\Jset\) be partitioned into real and dummy variables
\[
\Jset_R := \{1,\dots,p\}, \qquad
\Jset_D := \{p+1,\dots,p+L\}, \qquad
\Jset := \Jset_R \cup \Jset_D.
\]
At iteration \(k \ge 1\), \(\A_k \subseteq \Jset\) denotes the active set of selected variables, which is \(\F_k\)-measurable by construction. The next variable index is chosen from the inactive set via an \(\F_k\)-measurable decision rule
\[
\phi_k : \Omega \to\Jset \setminus \A_k,
\qquad
j^\star = \phi_k(\omega),
\]
and the active set is then updated recursively as
\(
\A_{k+1} = \A_k \cup \{j^\star\}.
\) Because the FS rule may select a dummy at any step $k\geq1$, we must formally capture
the iteration at which each dummy becomes fully revealed. Each dummy index
$\ell=1,\dots,L$ has an associated \emph{selection time}
\[
  \tau_\ell(\omega)
  := \inf\{\, k \ge 1 : \phi_k(\omega) = p+\ell \,\},
\]
indicating the iteration at which dummy~$\ell$ is first selected.
By construction, $\tau_\ell$ is a stopping time with respect to the filtration $(\F_k)_{k\ge0}$.\footnote{%
For any $k\ge1$, $\{\tau_\ell \le k\}=\bigcup_{i=1}^k\{\phi_i=p+\ell\}$. Since $\phi_i$ is $\F_i$-measurable and $\F_i\subseteq\F_k$ for $i\le k$, this event belongs to $\F_k$.} After selecting \(j^\star\) at step \(k\), we set \(\bm v_k\) to the corresponding variable and compute the next basis vector through orthonormalization:
\begin{equation}
\label{eq:orthog}
\bm v_{k}^\perp :=
  \bm v_{k} - \sum_{i=1}^k \langle \bm v_{k}, \bm e_i \rangle \bm e_i,
\qquad
\bm e_{k+1} := \frac{\bm v_{k}^\perp}{\|\bm v_{k}^\perp\|}.
\end{equation}
If the selected variable is a dummy with index
\(\ell^\star = \phi_k - p\),
its unrevealed residual is drawn conditionally on
\(\F_k\):
\begin{equation}
\label{eq:realize-dummy}
\bm d_{\ell^\star,k}^{\perp} \,\big|\, \F_k \sim \Psi_k^{\perp},
\qquad
\bm d_{\ell^\star}
  = \sum_{i=1}^k \upalpha_{i\ell^\star}\bm e_i
  + \bm d_{\ell^\star,k}^{\perp}.
\end{equation}
This operation \emph{realizes} dummy~\(\ell^\star\),
making \(\bm d_{\ell^\star}\) fully known. We then define the intermediate $\sigma$-algebra
\begin{equation}
\label{eq:common-filtration}
\F_k^+
:= \sigma\!\Bigl(\F_k,\ \{\bm d_\ell:\tau_\ell=k\}\Bigr),
\end{equation}
which includes any newly realized dummies and the new direction $\bm e_{k+1} = \bm d_{\ell^\star,k}^{\perp}/\|\bm d_{\ell^\star,k}^{\perp}\|$. Thus, \(\F_k^+\) represents the information after realizing any selected dummy but before drawing new projections. When a real variable is selected, $\bm e_{k+1}$ is already $\F_k$-measurable and no new dummies are realized, so $\F_k^+=\F_k$. When a dummy is selected, $\bm e_{k+1}$ becomes measurable only in $\F_k^+$. Conditional on $\F_k^+$, the next projections of all unrealized dummies (i.e., $\tau_\ell > k$) are drawn along the now fixed axis $\bm e_{k+1}$:
\begin{equation}
\label{eq:alpha-next}
\{\upalpha_{k+1,\ell}\}_{\tau_\ell>k} \;\big|\; \F_k^+
\stackrel{\text{i.i.d.}}{\sim} \Psi_{k+1},
\qquad
\upalpha_{k+1,\ell} = \langle \bm d_\ell, \bm e_{k+1}\rangle.
\end{equation}
We call these the \emph{fresh projections} at step $k+1$, as they have not yet entered the filtration. The filtration is then updated to
\begin{equation}
\label{eq:Fk-update}
\F_{k+1}
:= \sigma\!\Bigl(\F_k^+,\ \{\upalpha_{k+1,\ell}\}_{\tau_\ell>k}\Bigr).
\end{equation}
Hence, the timeline of information flow (illustrated in Figure~\ref{fig:filt_timeline}) is
\[
  \F_k
  \;\xrightarrow[\text{realize $\bm d_{\ell^\star}$}]{}
  \F_k^+
  \;\xrightarrow[\text{reveal $\{\upalpha_{k+1,\,\ell}\}_{\tau_\ell>k}$}]{}
  \F_{k+1}.
\]
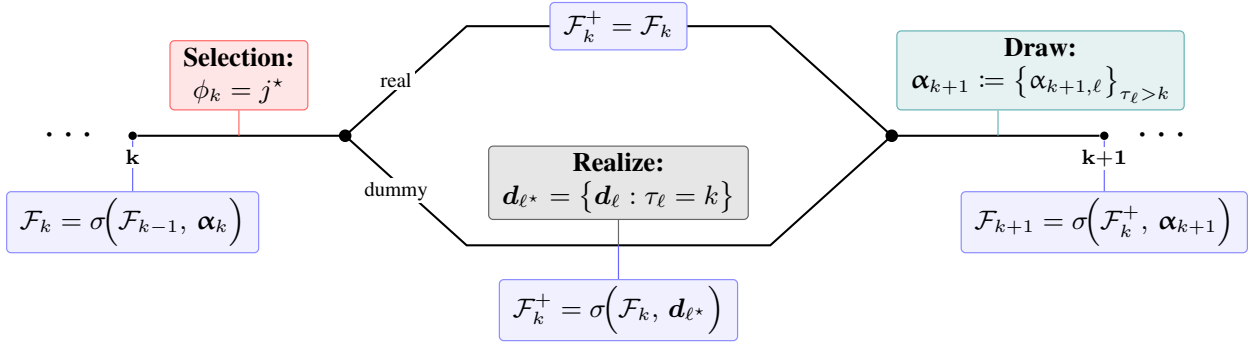
\begin{figure}[t]
    \centering
    \resizebox{\linewidth}{!}{%
\begin{tikzpicture}[
    >=latex,
    node distance=4mm and 8mm,
    axis/.style={line width=0.7pt},
    tick/.style={circle, fill=black, inner sep=1pt},
    branch dot/.style={circle, fill=black, inner sep=1.5pt},
    tlabel/.style={font=\scriptsize, fill=white, inner sep=1pt},
    box/.style={rounded corners=1.5pt, font=\footnotesize, inner xsep=4.5pt, inner ysep=2.2pt, align=center},
    timebox/.style={box, draw=blue!45, fill=blue!6, minimum height=1.6em, minimum width=3em},
    eventbox1/.style={box, draw=teal!60,  fill=teal!10,  minimum height=2.1em, minimum width=4em},
    eventbox2/.style={box, draw=red!60,   fill=red!10,   minimum height=2.1em, minimum width=4em},
    eventbox3/.style={box, draw=black!60, fill=black!10, minimum height=2.1em, minimum width=4em},
    branchlabel/.style={font=\scriptsize, fill=white, inner sep=1pt},
]
\def\sw{3.6}
\def\bdy{1.3}
\def\bslant{0.4}
% === MAIN AXIS: ticks at edges, dots immediately outside ===
\node at (-0.35*\sw, 0) {\Large $\cdots$};
\node[tick] (tk) at (-0.15*\sw,0) {};
\coordinate (fork)  at (0.55*\sw, 0);
\coordinate (merge) at (2.35*\sw, 0);
\node[tick] (tk1) at (3.05*\sw, 0) {};
\node at (3.25*\sw, 0) {\Large $\cdots$};
\draw[axis] (tk) -- (fork);
\draw[axis] (merge) -- (tk1);
\node[branch dot] at (fork) {};
\node[branch dot] at (merge) {};
% === UPPER BRANCH: j* ≤ p (rises) ===
\coordinate (up_entry) at ($(fork) +(\bslant*\sw, \bdy)$);
\coordinate (up_exit)  at ($(merge)+(-\bslant*\sw, \bdy)$);
\coordinate (up_mid)   at ($(up_entry)!0.5!(up_exit)$);
\coordinate (up_lab)   at ($(fork)!0.35!(up_entry)+(1mm,2.1mm)$);
\draw[axis] (fork) -- (up_entry) -- (up_exit) -- (merge);
\node[branchlabel] at (up_lab) {real};
% F_k^+ = F_k label on the upper horizontal
\node[timebox] at (up_mid) {$\F_k^+ = \F_k$};
% === LOWER BRANCH: j* > p (dips) ===
\coordinate (low_entry) at ($(fork) +(\bslant*\sw, -\bdy)$);
\coordinate (low_exit)  at ($(merge)+(-\bslant*\sw, -\bdy)$);
\coordinate (low_mid)   at ($(low_entry)!0.5!(low_exit)$);
\coordinate (low_lab)   at ($(fork)!0.35!(low_entry)+(1mm,-2mm)$);
\draw[axis] (fork) -- (low_entry) -- (low_exit) -- (merge);
\node[branchlabel] at (low_lab) {dummy};
% --- Realize: above the lower branch ---
\node[eventbox3, above=3mm of low_mid] (Real)
  {\textbf{Realize:}\\[1pt]
   $\bm d_{\ell^\star} = \big\{\bm d_\ell : \tau_\ell = k\big\}$};
\draw[-,black!60] (low_mid) -- (Real.south);
% --- F_k+: below lower branch ---
\node[timebox, below=4mm of low_mid] (Fkplus)
  {$\F_k^{+} = \sigma\!\Big(\F_k,\, \bm d_{\ell^\star}\Big)$};
\draw[-,blue!60] (low_mid) -- (Fkplus.north);
% === SELECTION: above main axis, between tk and fork ===
\coordinate (sel_line) at ($(tk)!0.5!(fork)$);
\node[eventbox2, above=3mm of sel_line] (Sel)
  {\textbf{Selection:}\\[1pt]
   $\phi_k = j^\star$};
\draw[-,red!60] (sel_line) -- (sel_line |- Sel.south);
% === DRAW: above main axis, between merge and tk1 ===
\coordinate (draw_line) at ($(merge)!0.5!(tk1)$);
\coordinate (draw_box)  at ($(merge)!0.7!(tk1)$);
\node[eventbox1, above=3mm of draw_box] (Draw)
  {\textbf{Draw:}\\[1pt]
   $\bm\upalpha_{k+1} \coloneqq \big\{\upalpha_{k+1,\ell}\big\}_{\tau_\ell > k}$};
\draw[-,teal!60] (draw_line) -- (draw_line |- Draw.south);
% === FILTRATION BOXES: below main axis ===
\node[timebox, below=6mm of tk] (Fk)
  {$\F_k = \sigma\!\Big(\F_{k-1},\, \bm\upalpha_k\Big)$};
\draw[-,blue!60] (tk) -- (Fk.north);
\node[timebox, below=6mm of tk1] (Fk1)
  {$\F_{k+1} = \sigma\!\Big(\F_k^{+},\, \bm\upalpha_{k+1}\Big)$};
\draw[-,blue!60] (tk1) -- (Fk1.north);
% === LABELS ===
\node[tlabel, below=2pt of tk]  {$\bf k$};
\node[tlabel, below=2pt of tk1] {$\bf k{+}1$};
\end{tikzpicture}%
}% end resizebox
    \caption{Exemplary timeline of the filtration $(\F_k)$. At step $k$, the selection rule $\phi_k$ is applied using only the information in $\F_k$. If a real variable is selected ($j^\star \le p$), the $\sigma$-algebra does not grow, hence $\F_k^+ = \F_k$. If a dummy is selected ($j^\star > p$), it is realized, yielding the intermediate $\sigma$-algebra $\F_k^+ = \sigma(\F_k,d_{\ell^\star})$. The next projections $\{\upalpha_{k+1,\ell}\}_{\tau_\ell>k}$ are then drawn conditional on $\F_k^+$, yielding $\F_{k+1}$.}
    \label{fig:filt_timeline}
\end{figure}
The filtration encodes exactly the information needed to make $\phi_k$ $\F_k$-measurable, and determines the conditional distribution of unrevealed dummy components.

\begin{rmrk}[Conditioning convention]
\label{rmrk:conditioning}
We state conditional laws of the next dummy projections given $\F_k^+$ rather than $\F_k$, since $\bm e_{k+1}$ is $\F_k^+$-measurable in all cases. When a real variable is selected, $\F_k^+=\F_k$ and no additional conditioning is introduced.
\end{rmrk}

Conditionally on $\F_k$, each dummy $\bm d_\ell$
admits the decomposition into an $\F_k$-measurable component and an
orthogonal residual,
\begin{align}
\label{eq:decomp_dummies}
\underbrace{
  \bm d_{\ell,k}^{\parallel}
  := \sum_{i=1}^k \upalpha_{i\ell}\bm e_i \in \V_k
}_{\text{$\F_k$-measurable}}
\qquad
\underbrace{
  \bm d_{\ell,k}^{\perp} \in H \cap \V_k^\perp,
  \quad \bm d_{\ell,k}^{\perp}\mid\F_k \sim \Psi_k^{\perp}
  \vphantom{\sum_{i=1}^k}
}_{\text{unrevealed residual}}
\end{align}
and for any $\ell$ with $\tau_\ell>k$, the conditional law is unchanged, as
$\bm d_{\ell,k}^{\perp}\mid\F_k^+\sim\Psi_k^\perp$,
since passing from $\F_k$ to $\F_k^+$ only reveals the dummy selected at step $k$ (if any) and does not reveal any information about the residual components of
the other, unrealized dummies.

The revealed projections define the affine slice
\[
\As_k :=
\Bigl\{\, \bm x \in H :
\langle \bm x, \bm e_i \rangle = \upalpha_{i\ell},\ i=1,\dots,k
\Bigr\}
= \bm d_{\ell,k}^{\parallel} + (H\cap\V_k^\perp).
\]

\noindent
With this structure in place, we can finally formalize the difference between explicit and virtual dummy treatments. The distinction is not in the update rule, but in which parts of the dummy ensemble are revealed at which times. We denote the two processes by
\[ \mathcal{P}_{\mathrm{ad}} \;\widehat{=}\; \text{Augmented Dummy Forward Selector},\qquad \mathcal{P}_{\mathrm{vd}} \;\widehat{=}\; \text{Virtual Dummy Forward Selector}.
\]

The augmented process~$\mathcal{P}_{\mathrm{ad}}$ draws the entire dummy ensemble $\bm D = (\bm d_1,\dots,\bm d_L)$ from $\Psi(H)^{\otimes L}$ \emph{a priori}.
Once $\bm D$ is fixed, its trajectory is deterministic and measurable with respect to the full $\sigma$-algebra $\sigma(\bm X,\bm y,\bm D)$. Thus $\mathcal{P}_{\mathrm{ad}}$ is generally \emph{not} adapted to the filtration $(\F_k)$; rather, $(\F_k)$ serves only as a device for comparing its information pattern to that of~$\mathcal{P}_{\mathrm{vd}}$, by conditioning on it.

In contrast, the virtual process~$\mathcal{P}_{\mathrm{vd}}$ reveals the same dummy ensemble incrementally through the filtration: at each step~$k$, it evolves from $\F_k$ to the intermediate $\F_k^+$ (realizing the selected dummy and updating the basis), and then to $\F_{k+1}$ (revealing the new projections $\{\upalpha_{k+1,\ell}\}$).
Hence $\mathcal{P}_{\mathrm{vd}}$ is adapted to $(\F_k)$. Both processes use identical update rules, basis evolution, and tie-breaking mechanisms; they differ only in how dummy randomness is revealed. Section~\ref{sec:stick_breaking} establishes that the two processes generate the same law over the entire forward-selection trajectory, i.e., the sequence of selected indices, basis vectors, correlations, and all intermediate path statistics.

Throughout, we assume that all orthonormalized candidates $\bm v_k^{\perp}$ and unrealized dummy residuals are almost surely nonzero; degenerate events (e.g.\ $\bm v_k \in \V_{k-1}$) occur with probability zero under continuous, rotationally invariant base laws. Accordingly, $\phi_k$ is $\F_k$-measurable at each step, and $\bm e_{k+1}$ is $\F_k$-measurable when a real variable is selected and $\F_k^+$-measurable when a dummy is selected.
\section{Sequential Sampling of Dummies}
\label{sec:stick_breaking}

In Section~\ref{sec:setup}, we introduced the probability space, the filtration, and an abstract forward--selection framework.  In this section, we specialize the dummies to rotationally invariant base laws $\Psi(H)$. This invariance is exactly what allows us to sample fresh dummy projections conditionally on the revealed history \((\F_k^+)_{k\ge0}\), without generating the full dummy matrix $\bm D$.

\subsection{The Role of Rotational Invariance}
\label{ssec:rotational-invariance}

The central geometric property we rely on is rotational invariance. Its relevance comes from the two ways a dummy vector can be represented: (i) via fixed coordinates in the canonical basis of $\R^n$ (augmented selection), or (ii) through projections $\{\upalpha_{j\ell}\}$ onto the \emph{adaptive basis} $\Ec_k$ (virtual selection). For these representations to produce identically distributed paths, the law $\Psi(H)$ must be invariant under any orthogonal transformation $\bm Q \in \mathrm{O}(H)$.

\begin{lem}[Gaussian rotational invariance]
\label{lem:rot-invariance}
Let $\bm D\in\R^{n\times L}$ have i.i.d.\ rows drawn from
$\mathcal N_L(\bm 0,\bm\Sigma)$ with $\bm\Sigma\succ 0$. Then
$\bm Q\bm D \stackrel{d}{=} \bm D$ for any $\bm Q\in\mathrm O(n)$.
In particular, each column of $\bm D$ is marginally rotationally
invariant in $\R^n$, even when distinct columns are correlated
through~$\bm\Sigma$.
\end{lem}
\begin{proof}
The proof is deferred to Appendix~\ref{app:B}.
\end{proof}
This standard result shows that Gaussian row sampling is invariant under
left-orthogonal transformations of $\R^n$, regardless of the dependence
structure among the $L$ columns. Moreover, if \(\Psi=\mu^{\otimes n}\) is a product measure, Herschel--Maxwell's theorem, which states that the Gaussian is the only distribution that is both rotationally symmetric and has independent coordinates, implies that mean-zero Gaussian distributions are the \emph{only} possibility. Thus, requiring exact rotational invariance for i.i.d.\ dummies forces the choice of a Gaussian base law. This leads to two canonical choices, conforming with the T-Rex selector, which serve as our main examples:

\begin{itemize}
\item[\textbf{(\,G\,)}]
  $\Psi(H) = \mathcal N(\bm 0, \bm I_H)$, where $\bm I_H := \bm I_n - \frac{1}{n}\one\one^\top$. 
  This defines a standard Gaussian on the subspace $H$.\footnote{
  Equivalently, if $\{\bm f_1, \dots, \bm f_m\}$ is any orthonormal basis for $H$,
  then $\bm d_\ell = \sum_{j=1}^m z_j \bm f_j$ with $z_j \stackrel{\text{i.i.d.}}{\sim} \mathcal{N}(0,1)$.}
\item[\textbf{(\:S\:)}]
  $\Psi(H) = \Unif(S_H)$, where \(S_H := \{\bm x \in H : \|\bm x\|_2 = 1\}\). This defines the uniform distribution on the unit sphere $S_H$ centered in $H$. Equivalently, let $\bm g_\ell \sim \mathcal N(\bm 0, \bm I_H)$ and normalize $\bm d_\ell = \bm g_\ell / \|\bm g_\ell\|_2$.
\end{itemize}

Both laws are rotationally invariant, but they differ in dependence: the Gaussian law yields independent coordinates, while the spherical law couples them through the unit-norm constraint. While (\,G\,) provides a simple baseline, the standardization step in most forward-selection methods maps columns onto \(S_H\), making the spherical model the canonical setting for sequential sampling.

\subsection{Adaptive Stick-Breaking}

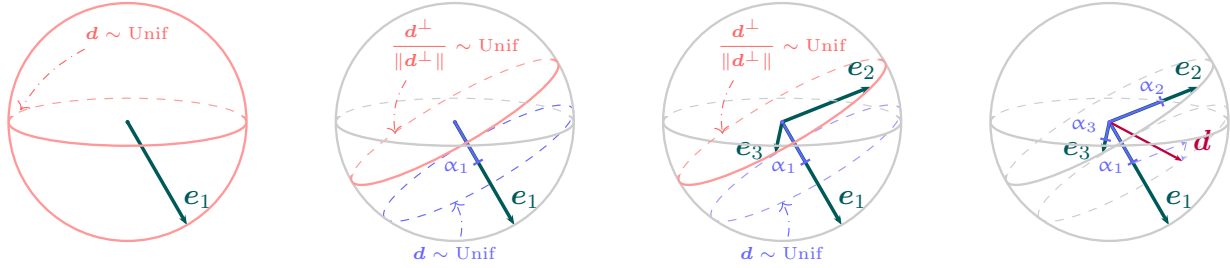
\begin{figure}
    \centering
    \newcommand{\subsphereoutline}[5][]{%
    % back half (hidden)
    \draw[#1, dashed, rotate=#5] (#2) circle [x radius=#3, y radius=#4];
    % front half (visible)
    \begin{scope}[rotate=#5]
        \clip (#3,0.1) rectangle (-#3,-#4-0.1);
        \draw[#1, thick] (#2) circle [x radius=#3, y radius=#4];
    \end{scope}
}
\resizebox{\linewidth}{!}{
\begin{tikzpicture}[scale=0.72]

    % --- FIGURE 1: INITIAL STATE (Left) ---
    \begin{scope}[shift={(0,0)}]
        \def\R{2} \def\tilt{-60} \def\yPersp{0.4}
        \fill[teal!70!black] (0,0) circle (1pt);
        \draw[-{Stealth[scale=0.4]}, teal!70!black, very thick] (0,0) -- (\tilt:\R) node[pos=0.75, right] {$\bm{e}_1$};
        \draw[red!40, thick] (0,0) circle (\R);
        \subsphereoutline[red!40]{0,0}{\R}{\yPersp}{0};
        \node[red!60, font=\tiny] (lab1) at (0,1.5) {$\bm d \sim \mathrm{Unif}$};
        \draw[red!60, ->, dashdotted] (lab1) to[bend right=15] (-1.8,0.2);
    \end{scope}

    % --- FIGURE 2: FIRST PROJECTION (Center-Left) ---
    \begin{scope}[shift={(5.5,0)}]
        \def\R{2} \def\alphaVal{0.8} \def\tilt{-60} \def\subR{1.83} \def\yPersp{0.4}
        \fill[blue!60] (0,0) circle (1pt);
        \draw[-{Stealth[scale=0.4]}, teal!70!black, very thick] (0,0) -- (\tilt:\R) node[pos=0.75, right] {$\bm{e}_1$};
        \coordinate (Center) at (\tilt:\alphaVal);
        \subsphereoutline[blue!60]{Center}{\subR}{\yPersp}{\tilt-90}
        \draw[blue!60, thick, -{|[scale=0.6]}] (0.0,0) -- (Center) node[pos=1.1, left, font=\scriptsize] {$\upalpha_1$};
        \draw[gray!40, thick] (0,0) circle (\R);
        \subsphereoutline[red!40]{0,0}{\R}{\yPersp}{\tilt-270};
        \subsphereoutline[gray!40]{0,0}{\R}{\yPersp}{0};
        
        \node[red!60, font=\tiny, align=left] (lab2) at (0,1.3)
          {$\dfrac{\bm d^\perp}{\|\bm d^\perp\|} \sim \mathrm{Unif}$};
        \draw[red!60, ->, dashdotted]
          (lab2) to[bend right=20] (-1,-0.2);

         \node[blue!60, font=\tiny, align=left] (lab2) at (0,-2.2)
          {${\bm d} \sim \mathrm{Unif}$};
        \draw[blue!60, ->, dashdotted]
          (lab2) to[bend right=20] (0,-1.35);
    \end{scope}

    % --- FIGURE 3: SUBSEQUENT DIRECTIONS (Center-Right) ---
    \begin{scope}[shift={(11,0)}]
        \def\R{2} \def\alphaVal{0.8} \def\tilt{-60} \def\subR{1.83} \def\yPersp{0.4}
        \pgfmathsetmacro{\rotAng}{\tilt+90}
        \pgfmathsetmacro{\cA}{cos(\rotAng)} \pgfmathsetmacro{\sA}{sin(\rotAng)}
        \pgfmathsetmacro{\eTwoX}{cos(-20)} \pgfmathsetmacro{\eTwoY}{\yPersp*sin(-20)}
        \pgfmathsetmacro{\eTwoN}{sqrt(\eTwoX*\eTwoX + \eTwoY*\eTwoY)}
        \pgfmathsetmacro{\eThreeX}{cos(250)} \pgfmathsetmacro{\eThreeY}{\yPersp*sin(250)}
        \pgfmathsetmacro{\eThreeN}{sqrt(\eThreeX*\eThreeX + \eThreeY*\eThreeY)}
        \coordinate (E1u) at (\tilt:1);
        \coordinate (E2u) at ({(\eTwoX/\eTwoN)*\cA - (\eTwoY/\eTwoN)*\sA}, {(\eTwoX/\eTwoN)*\sA + (\eTwoY/\eTwoN)*\cA});
        \coordinate (E3u) at ({(\eThreeX/\eThreeN)*\cA - (\eThreeY/\eThreeN)*\sA}, {(\eThreeX/\eThreeN)*\sA + (\eThreeY/\eThreeN)*\cA});
        
        \draw[-{Stealth[scale=0.4]}, teal!70!black, very thick] (0,0) -- ($(0,0) + \R*(E1u)$) node[pos=0.75, right] {$\bm e_1$};
        \coordinate (Center) at ($(0,0) + \alphaVal*(E1u)$);
        \subsphereoutline[blue!40]{Center}{\subR}{\yPersp}{\tilt-90}
        \draw[blue!60, thick, -{|[scale=0.6]}] (0,0) -- (Center) node[pos=1.1, left, font=\scriptsize] {$\upalpha_1$};
        
        \draw[-{Stealth[scale=0.4]}, teal!70!black, very thick] (0,0) -- ($(0,0) + 0.8*\R*(E2u)$) node[pos=0.9, above] {$\bm e_2$};
        \draw[-{Stealth[scale=0.4]}, teal!70!black, very thick] (0,0) -- ($(0,0) + 0.27*\R*(E3u)$) node[pos=0.9, left] {$\bm e_3$};
        
        \fill[blue!60] (0,0) circle (1pt);
        \draw[gray!40, thick] (0,0) circle (\R);
        \subsphereoutline[gray!40]{0,0}{\R}{\yPersp}{0};
        \subsphereoutline[red!40]{0,0}{\R}{\yPersp}{\tilt-270};
        \node[red!60, font=\tiny, align=left] (lab2) at (0,1.3)
          {$\dfrac{\bm d^\perp}{\|\bm d^\perp\|} \sim \mathrm{Unif}$};
        \draw[red!60, ->, dashdotted]
          (lab2) to[bend right=20] (-1,-0.2);
        \node[blue!60, font=\tiny, align=left] (lab2) at (0,-2.2)
          {${\bm d} \sim \mathrm{Unif}$};
        \draw[blue!60, ->, dashdotted]
          (lab2) to[bend right=20] (0,-1.35);
    \end{scope}

    % --- FIGURE 4: VECTOR CONSTRUCTION (Right) ---
    \begin{scope}[shift={(16.5,0)}]
        \def\R{2} \def\alphaVal{0.8} \def\tilt{-60} \def\subR{1.83} \def\yPersp{0.4}
        \def\alphaTwo{0.6} \def\alphaThree{0.6}
        \pgfmathsetmacro{\rotAng}{\tilt+90}
        \pgfmathsetmacro{\cA}{cos(\rotAng)} \pgfmathsetmacro{\sA}{sin(\rotAng)}
        \pgfmathsetmacro{\eTwoX}{cos(-20)} \pgfmathsetmacro{\eTwoY}{\yPersp*sin(-20)}
        \pgfmathsetmacro{\eTwoN}{sqrt(\eTwoX*\eTwoX + \eTwoY*\eTwoY)}
        \pgfmathsetmacro{\eThreeX}{cos(250)} \pgfmathsetmacro{\eThreeY}{\yPersp*sin(250)}
        \pgfmathsetmacro{\eThreeN}{sqrt(\eThreeX*\eThreeX + \eThreeY*\eThreeY)}
        \coordinate (E1u) at (\tilt:1);
        \coordinate (E2u) at ({(\eTwoX/\eTwoN)*\cA - (\eTwoY/\eTwoN)*\sA}, {(\eTwoX/\eTwoN)*\sA + (\eTwoY/\eTwoN)*\cA});
        \coordinate (E3u) at ({(\eThreeX/\eThreeN)*\cA - (\eThreeY/\eThreeN)*\sA}, {(\eThreeX/\eThreeN)*\sA + (\eThreeY/\eThreeN)*\cA});
        
        \draw[-{Stealth[scale=0.4]}, teal!70!black, very thick] (0,0) -- ($(0,0) + \R*(E1u)$) node[pos=0.75, right] {$\bm e_1$};
        \coordinate (Center) at ($(0,0) + \alphaVal*(E1u)$);
        \subsphereoutline[gray!40]{Center}{\subR}{\yPersp}{\tilt-90}
        \draw[blue!60, thick, -{|[scale=0.6]}] (0,0) -- (Center) node[pos=1.1, left, font=\scriptsize] {$\upalpha_1$};
        
        \draw[-{Stealth[scale=0.4]}, teal!70!black, very thick] (0,0) -- ($(0,0) + 0.8*\R*(E2u)$) node[pos=0.9, above] {$\bm e_2$};
        \draw[-{Stealth[scale=0.4]}, teal!70!black, very thick] (0,0) -- ($(0,0) + 0.27*\R*(E3u)$) node[pos=0.9, left] {$\bm e_3$};
        
        \coordinate (AlphaTwoPt) at ($(0,0) + 0.8*\R*\alphaTwo*(E2u)$);
        \coordinate (AlphaThreePt) at ($(0,0) + 0.27*\R*\alphaThree*(E3u)$);
        \draw[blue!60, thick, -{|[scale=0.6]}] (0,0) -- (AlphaTwoPt) node[pos=0.8, above, font=\scriptsize] {$\upalpha_2$};
        \draw[blue!60, thick, -{|[scale=0.6]}] (0,0) -- (AlphaThreePt) node[pos=0.2, left, font=\scriptsize] {$\upalpha_3$};
        
        \coordinate (P1) at (Center);
        \coordinate (P2) at ($(P1) + 0.8*\R*\alphaTwo*(E2u)$);
        \coordinate (D)  at ($(P2) + 0.27*\R*\alphaThree*(E3u)$);
        \draw[blue!40, dashed,-{Stealth[scale=0.4]}] (P1) -- (P2);
        \draw[blue!40, dashed, -{Stealth[scale=0.4]}] (P2) -- (D);
        \draw[-{Stealth[scale=0.4]}, purple, thick] (0,0) -- (D) node[anchor=south west] {$\bm d$};
        
        \fill[blue!60] (0,0) circle (1pt);
        \draw[gray!40, thick] (0,0) circle (\R);
        \subsphereoutline[gray!40]{0,0}{\R}{\yPersp}{0};
        \subsphereoutline[gray!40]{0,0}{\R}{\yPersp}{\tilt-270};
    \end{scope}

\end{tikzpicture}
}

    \caption{
    Illustration of the stick-breaking construction for spherical dummies.
    From left to right: Initially, the dummy $\bm d$ is uniformly distributed on the sphere $S_H$, and the first direction $\bm e_1$ is revealed.
    Drawing $\upalpha_1$ constrains the normalized residual direction $\bm d^\perp/\|\bm d^\perp\|$ to lie uniformly on the red subsphere, while the dummy $\bm d$ itself becomes uniformly distributed on the corresponding affine slice (blue).
    The FS algorithm then returns the next direction $\bm e_2$ (which in three dimensions fixes $\bm e_3$ up to sign) and draws $\upalpha_2$. The magnitude of $\upalpha_3$ is fixed and its sign drawn uniformly at random. At this point all randomness is exhausted, and the dummy $\bm d$ is fully determined, i.e., $\F_k$-measurable.
}
    \label{fig:stick_breaking}
\end{figure}
In our framework, the orthonormal system $\{\bm e_k\}$ is genuinely \emph{adaptive}, evolving with the data. This raises a crucial question: \emph{does the conditional uniformity of the unrealized dummy components persist when the basis itself is a random, filtration-measurable object?} 

We show that this conditional symmetry does indeed survive. Rotational invariance on $S_H$ propagates through the filtration, ensuring that at each step, the residual component of any unrealized dummy remains uniform on the orthogonal complement of the currently revealed subspace (equivalently, the nested sequence $\V_0\subset\V_1\subset\cdots$ defines a Markov chain on the Grassmannian $\mathrm{Gr}(k,H)$, though we do not develop this perspective further). This stability allows us to translate the abstract symmetry into a recursive sampling scheme where, at each step $k$, the unrevealed portion of a dummy $\bm d_\ell$ is simply restricted to a smaller, lower-dimensional sphere. This is formalized in the following lemma.

\begin{lem}[Conditional Uniformity]
\label{lem:conditional-uniformity-affine}
Fix $k\in\{0,\dots,m-1\}$ and $\ell\in\{1,\dots,L\}$ and assume $\tau_\ell>k$ (i.e., dummy $\ell$ has not been selected yet at iteration k). Denote $\upalpha_i := \upalpha_{i\ell}$.
With the common filtration $\F_k^+$ defined in \eqref{eq:common-filtration}, the random vector $\bm d_\ell$ satisfies, conditionally on $\F_k^+$,
\[
\bm d_\ell \;\big|\; \F_k^+ \;\sim\; \Unif\bigl(S_H\cap \As_k\bigr),
\]
where $\As_k := \{ \bm x \in H : \langle \bm x,\bm e_i\rangle = \upalpha_i,\ i=1,\dots,k\}$. Equivalently, the residual $\bm d_\ell^\perp := \bm d_\ell - \sum_{i=1}^k \upalpha_i \bm e_i$ has direction uniformly distributed on $S_H\cap\V_k^\perp$ and radius $R_{k}=\sqrt{1-\sum_{i=1}^k\upalpha_i^2}$.
\end{lem}
\begin{proof}
The proof is deferred to Appendix~\ref{app:B}.
\end{proof}

Lemma~\ref{lem:conditional-uniformity-affine} establishes that each unrealized dummy remains uniform on a subsphere in $H$ despite the adaptive basis. In a fixed basis, it is a classical result that the squared coordinates of such a vector follow a symmetric Dirichlet law and can be generated via stick-breaking. We now show that this Dirichlet structure is preserved under the filtration $(\F_k)$, enabling the sequential generation of the path. To that end, let us first briefly review the finite-dimensional stick-breaking construction of the Dirichlet distribution.

\paragraph{Stick-Breaking in a Fixed Basis}

Let $\bm d \sim \Unif(S_H)$ and let $\{\bm f_1,\dots,\bm f_m\}$ be any fixed orthonormal basis of $H$. Define
\[
\upalpha_k := \langle \bm d, \bm f_k\rangle,
\qquad
Z_k := \upalpha_k^2,
\qquad
\sum_{k=1}^m Z_k = 1.
\]
It is well known that the squared coordinates of a uniform spherical vector follow a Dirichlet law,
\[
(Z_1,\dots,Z_m)
\sim \mathrm{Dirichlet}\!\left(\tfrac12,\dots,\tfrac12\right).
\]

Such a Dirichlet vector is \emph{neutral} \citep{frigyik2010,Connor1969}: for every $k<m$, the normalized remainder satisfies 
\[
\frac{(Z_{k+1},\dots,Z_m)}
     {1-\sum_{i=1}^k Z_i}~\sim~
\mathrm{Dirichlet}
\bigl(\tfrac12,\dots,\tfrac12\bigr)
\quad\text{on the $(m-k)$-simplex (dimension $m-k-1$)}
\]
and is independent of $(Z_1,\dots,Z_k)$.
Aggregating the last $m-k$ coordinates of this Dirichlet distribution\footnote{ We use the identity that a vector $(W_1,\dots,W_r)\sim\mathrm{Dirichlet}(\zeta_1,\dots,\zeta_r)$ satisfies~$W_i\sim~\mathrm{Beta}(\zeta_i,\sum_{\ell\neq i}\zeta_\ell)$.}
yields the conditional marginal
\[
\frac{Z_{k+1}}{1-\sum_{i=1}^k Z_i}
\,\Big|\,(Z_1,\dots,Z_k)
~\sim~
\mathrm{Beta}\!\left(\tfrac12,\tfrac{m-k-1}{2}\right),
\qquad k=1,\dots,m-2.
\]

Thus the coordinates admit the classical \emph{finite stick-breaking} representation:
\[
Z_1 = V_1,\qquad
Z_2 = (1-V_1)V_2,\qquad
\dots,\qquad
Z_m = \prod_{i=1}^{m-1}(1-V_i),
\]
where the independent breaking proportions satisfy
\[
V_i \sim \mathrm{Beta}\!\left(\tfrac12,\tfrac{m-i}{2}\right),
\qquad i=1,\dots,m-1.
\]
Finally, the original coordinates satisfy $\upalpha_k = S_k \sqrt{Z_k}$, where $S_k \in \{\pm1\}$ are i.i.d.\ Rademacher signs ($\Pr(S_k = +1) = \Pr(S_k = -1) = 1/2$).

\medskip
This provides the canonical link between the \(\mathrm{Dirichlet}(\tfrac12,\dots,\tfrac12)\) distribution and a sequence of independent Beta marginals, the probabilistic foundation for the adaptive stick-breaking construction developed below.

\begin{lem}
\label{lem:adaptive_stick_breaking}
Fix $\ell\in\{1,\dots,L\}$ and let $(\bm e_1,\dots,\bm e_m)$ be an orthonormal basis of $H$
such that for each $k\in\{0,\dots,m-1\}$ the prefix $(\bm e_1,\dots,\bm e_k)$ is $\F_k^+$-measurable.
Define
\[
Z_i := \upalpha_{i\ell}^2=\langle \bm d_\ell,\bm e_i\rangle^2,\qquad i=1,\dots,m.
\]
Assume $\bm d_\ell \notin \V_k$ for all $k<m$.

\begin{enumerate}
  \item $(Z_1,\dots,Z_m)\sim \mathrm{Dirichlet}(\tfrac12,\dots,\tfrac12)$.
  \item For any $k\in\{0,\dots,m-1\}$, on the event $\{\tau_\ell>k\}$,
  \[
  (Z_{k+1},\dots,Z_m)\mid \F_k^+
  ~\sim~ \Bigl(1-\sum_{i=1}^k Z_i\Bigr)\,
  \mathrm{Dirichlet}\!\left(\tfrac12,\dots,\tfrac12\right).
  \]
\end{enumerate}
\end{lem}
\begin{proof}
The proof is deferred to Appendix~\ref{app:B}.
\end{proof}
\begin{rmrk}[Stick-breaking and the adaptive twist]
The classical fact is that if $\bm d\sim\Unif(S_H)$ and
$\Ec_m$ is a fixed orthonormal basis, then $(\upalpha_1^2,\dots,\upalpha_m^2)$ is Dirichlet$(\tfrac12,\dots,\tfrac12)$, equivalently obtained by the standard stick-breaking scheme. Lemma~\ref{lem:adaptive_stick_breaking} shows this Dirichlet structure survives \emph{adaptive} basis selection in \(H\): even when $\bm e_{k+1}$ is chosen measurably from $\F_k^+$, the residual direction is uniform on the lower-dimensional sphere in \(H\), so the next squared projection is Beta-distributed
and the remaining ``mass'' continues to break according to the same Dirichlet law. This justifies using the same sequential Beta sampling in Algorithm~\ref{alg:seq_spherical} despite the data-dependent basis.
\end{rmrk}

\begin{algorithm}[ht]
\caption{Adaptive stick-breaking}
\label{alg:seq_spherical}
\begin{algorithmic}[1]
  \State $R_0^2 \gets 1$.
  \For{$k = 1$ to $m-1$}
    \State Draw independently of $\F_{k-1}^+$:
      \[
        U_k \sim \mathrm{Beta}\Bigl(\tfrac12,\tfrac{m-k}{2}\Bigr),
        \qquad
        S_k \sim \{\pm1\} \text{ (each prob.\ }1/2).
      \]
    \State Set
      \[
        \upalpha_k \gets S_k\sqrt{\,R_{k-1}^2\,U_k\,},
        \qquad
        R_k^2 \gets R_{k-1}^2(1-U_k).
      \]
    \State Reveal $\upalpha_k=\langle \bm d,\bm e_k\rangle$ (so it is $\F_k$--measurable).
  \EndFor
  \State Draw $S_m \sim \{\pm1\}$ and set $\upalpha_m \gets S_m\,R_{m-1}$.
  \State \Return $\bm d = \sum_{k=1}^m \upalpha_k\,\bm e_k$.
\end{algorithmic}
\end{algorithm}

\begin{cor}
\label{cor:alg_dist}
Let \((\bm e_1, \dots, \bm e_m)\) be any orthonormal basis of \(H\), possibly chosen adaptively. Let \((\upalpha_1, \dots, \upalpha_m)\) be the sequence of coefficients constructed according to Algorithm~\ref{alg:seq_spherical}. Then
\(
\bm d := \sum_{i=1}^m \upalpha_i\,\bm e_i
\)
is distributed uniformly on \(S_H\subset H\), i.e.\ \(\bm d \sim \Unif(S_H)\).
\end{cor}
\begin{proof}
The proof is deferred to Appendix~\ref{app:B}.
\end{proof}
Corollary~\ref{cor:alg_dist} establishes that Algorithm~\ref{alg:seq_spherical}  generates dummy vectors with the correct spherical distribution, regardless of  how the orthonormal basis is constructed from past data. In particular, each virtual dummy produced by the sequential sampler has exactly the same distribution as an explicitly drawn dummy, and the sampler reveals fresh projections in the same order and with the same conditional laws as the augmented procedure.

This alignment of conditional distributions is precisely what is required for the virtual-dummy forward selector to mimic the behavior of the explicitly augmented version. We now show that the entire forward-selection trajectories of the two procedures agree in distribution.

\begin{thm}[Distributional equivalence of VD--FS and AD--FS]
\label{thm:vd-fs-dist}
Let $\bm d_1,\dots,\bm d_L \stackrel{\mathrm{i.i.d.}}{\sim}\Psi(H)$ be dummy 
vectors independent of $(\bm X,\bm y)$, where $\Psi(H)$ is $\mathcal{N}(\bm 0, \bm I_H)$ or $ \Unif(S_H)$.
Let $\mathcal{P}_{\mathrm{ad}}$ denote the augmented-dummy forward selector 
(AD--FS), and $\mathcal{P}_{\mathrm{vd}}$ the virtual-dummy forward selector (VD--FS), 
both using the same deterministic update rule and tie-breaking.
Suppose that the forward-selection rule $\phi_k$ is $\F_k$-measurable at each step. Then the two procedures generate trajectories with identical probability law:
\[
\mathcal{P}_{\mathrm{ad}} \stackrel{d}{=} \mathcal{P}_{\mathrm{vd}}.
\]
\end{thm}
\begin{proof}
The proof is deferred to Appendix~\ref{app:B}.
\end{proof}
The distributional equivalence established in Theorem~\ref{thm:vd-fs-dist} implies that, for any fixed design $(\bm X,\bm y)$, the virtual-dummy procedure and the explicitly augmented procedure generate identical conditional laws for all quantities produced by the forward-selection path: correlations, step statistics, active sets, and stopping times. This becomes useful when integrated in inference methods that rely on the competition of real and dummy variables in order to control the FDR. 
The virtual-dummy construction is not tied to a particular selector. The only requirement imposed so far is that the selection rule at step $k$ be $\F_k$-measurable. This naturally includes a broad class of forward-selection procedures. A detailed compatibility discussion, together with examples of selectors that are and are not compatible, is given in Appendix~\ref{app:compatible}. In Section~\ref{sec:Var_sel_FDR}, we show how this transfers to LARS and, more generally, to the T--Rex methodology.
\section{Asymptotic and Finite Sample Analysis}
\label{sec:universality}
In this section, we complement the exact finite--sample equivalence of Section~\ref{sec:stick_breaking} (which requires rotationally invariant dummy laws) with an asymptotic result for generic standardized i.i.d.\ dummy constructions. Specifically, we prove a \emph{pathwise universality} statement in the regime where the sample size $n\to\infty$ while the number of forward-selection steps $K$ (and the pool size $L$ and ambient dimension $p$) are held fixed: the joint law of the first $K$ selected indices is asymptotically the same as if the dummies were Gaussian.
Crucially, this universality holds despite the fact that the FS directions are adaptive and depend on the entire selection history. Because the basis vectors $\bm e_k$ are random, data--dependent, and constructed sequentially, the proof necessarily relies on conditional limit theorems and stability properties of the selection rule. We therefore collect the standing assumptions under which the universality argument is valid.

\begin{assumption}
\label{ass:universality}
Fix $K, L, p \ge 1$ (independent of $n$). For each $n$, let $\bm\updelta^{(n)}_1,\dots,\bm\updelta^{(n)}_L\in\R^n$ be independent raw dummy vectors with i.i.d.\ coordinates $\updelta_{i\ell}^{(n)}$ satisfying $\E[\updelta_{i\ell}^{(n)}]=0$ and $\E[(\updelta_{i\ell}^{(n)})^2]=1$. Define the (normalized) dummy predictors in $H$ by
\[
  \bm d_\ell^{(n)}
  :=
  \frac{I_H\bm\updelta_\ell^{(n)}}{\|I_H\bm\updelta_\ell^{(n)}\|_2/\sqrt n}
  \;\in\;H,
  \qquad \ell=1,\dots,L.
\]
Let the forward selector generate basis vectors $\bm e_1^{(n)},\dots,\bm e_K^{(n)}$ and selection times $\tau_\ell^{(n)}$ as in Section~\ref{ssec:filtration}. We assume:
\begin{enumerate}
\item[\textbf{(D)}] \textbf{Delocalization:}
\[
  \max_{1\le k\le K}\|\bm e_k^{(n)}\|_\infty
  \xrightarrow[n\to\infty]{P} 0.
\]

\item[\textbf{(GP)}] \textbf{General position:} No correlation ties occur in the first $K$ steps with probability one.

\item[\textbf{(Sel)}] \textbf{Finite--set stability.}
For any $k\in\{0,\dots,K\}$ and any finite $\Lambda\subset\{1,\dots,L\}$, consider running the FS procedure on the augmented design with all dummies, and on the augmented design obtained by removing the dummies $\{\bm d_\ell^{(n)}:\ell\in\Lambda\}$, using the same $(\bm X,\bm y)$ and tie--breaking. If $\tau_\ell^{(n)} > k$ for all $\ell\in\Lambda$, then the two trajectories coincide up to step $k$ and yield the same $\bm e_{k+1}^{(n)}$.
\end{enumerate}
\end{assumption}

The conditions collected in Assumption~\ref{ass:universality}  are standard in high-dimensional analyses of greedy and pathwise procedures.

The delocalization condition \textbf{(D)} ensures that the data-dependent directions $\bm e_k^{(n)}$ do not concentrate on a vanishing number of coordinates; it is satisfied, for example, when the design and residuals are sufficiently spread out, and is typical in random design or high-dimensional regimes. The general position condition \textbf{(GP)} rules out exact correlation ties along the first $K$ steps. Under continuous dummy distributions, this holds with probability one and ensures that the selection rule is almost surely continuous with respect to the dummy projections. Finally, the finite-set stability condition \textbf{(Sel)} is a mild structural property of forward-selection algorithms: it states that dummies which are never selected cannot influence the selection path or the induced basis. This condition is satisfied by all standard forward-selection procedures, including LARS, OMP, and forward stepwise regression, since their updates depend only on the current residual and the selected variables. We now formalize this universality beyond exactly rotationally invariant Gaussian and spherical base laws. For each $n$, let
\[
\mathscr J_\mu^{(n)} := (J_{\mu,1}^{(n)},\dots,J_{\mu,K}^{(n)})
\] 
denote the sequence of indices selected in the first $K$ steps of the forward-selection procedure when the dummy variables are constructed from i.i.d.\ standardized coordinates with law $\mu$ under the conditions specified in Assumption~\ref{ass:universality}. We denote by $\mathscr J_{\mathrm G}^{(n)} := (J_{\mathrm G,1}^{(n)},\dots,J_{\mathrm G,K}^{(n)})$ the special case where $\mu = \mathcal{N}(0,1)$.

\begin{thm}[Pathwise universality]
\label{thm:pathwise-univ}
Under Assumption~\ref{ass:universality}, let 
\[
  \mathscr J_\mu^{(n)}
  := (J_{\mu,1}^{(n)},\dots,J_{\mu,K}^{(n)})
\]
be the selection path generated by i.i.d.\ non-Gaussian dummies, and let $\mathscr J_{\mathrm G}^{(n)}$ denote the path obtained with Gaussian  dummies (using the same FS rule, same $(\bm X,\bm y)$,  and Gaussian base law on $H$). Then there exists a random path 
$\mathscr J=(J_1,\dots,J_K)$ such that
\[
  \mathscr J_\mu^{(n)} \xRightarrow[n\to\infty]{d} \mathscr J,
  \qquad
  \mathscr J_{\mathrm G}^{(n)} \xRightarrow[n\to\infty]{d} \mathscr J.
\]
\end{thm}

\begin{proof}
The proof is deferred to Appendix~\ref{app:B}. The idea of the proof is sketched below.
\end{proof}

Theorem~\ref{thm:pathwise-univ} asserts that, for any fixed number of FS steps $K$, the entire selection path, i.e., the ordered sequence of selected indices, has the same asymptotic law under generic standardized i.i.d.\ dummies as under Gaussian dummies. This statement is pathwise: it concerns the joint distribution of $(J_1,\dots,J_K)$ and not merely marginal test statistics or single-step behavior. At a high level, the result rests on the fact that forward selection interacts with each dummy only through a growing collection of low-dimensional projections onto the data-dependent directions $\bm e_k$. Under the delocalization assumption~\textbf{(D)}, each fresh projection is an average of many weakly weighted coordinates. As a consequence, conditional on the current history, fresh dummy projections satisfy a conditional CLT (Lemma~\ref{lem:conditional-clt}), and the resulting normalized projections of any finite set of unselected dummies are asymptotically i.i.d.\ Gaussian (Lemma~\ref{lem:finite-dummy}). Pushing this from finite subsets to the entire pool of remaining dummies yields the Gaussian limit for the fresh projection vector at step $k{+}1$ conditional on $\F_k^+$ (Lemma~\ref{lem:fresh-step}).

The proof of Theorem~\ref{thm:pathwise-univ} proceeds by induction on the step index $k$. The induction step uses Lemma~\ref{lem:fresh-step} to replace the conditional law of the next dummy projections under $\mu$ by the corresponding Gaussian conditional law, given the same history. The finite-set stability assumption \textbf{(Sel)} enters through Lemma~\ref{lem:finite-dummy} and is subsequently used in Lemma~\ref{lem:fresh-step} to control the conditional law of fresh dummy projections. It ensures that dummies which are not selected up to step $k$ do not affect the evolution of the adaptive directions, so the basis can be treated as fixed when conditioning on a leave--$\Lambda$--out environment. General position~\textbf{(GP)} then guarantees that the selection map is almost surely continuous, allowing one to pass from convergence of projections to convergence of selected indices via a conditional continuous mapping argument. Unlike classical universality results for maxima or fixed test statistics, the present argument must control an adaptive, history-dependent sequence of decisions. The key novelty is that Gaussian behavior is recovered conditionally at each step in the evolving filtration, which is what allows the full forward-selection trajectory to converge in law.

\begin{rmrk}[On $L\to\infty$]
Theorem~\ref{thm:pathwise-univ} is proved in a fixed-$(K,L,p)$ asymptotic regime. This is natural for universality with non-rotationally-invariant i.i.d.\ dummies, because extending the argument to $L=L_n\to\infty$ would require uniform conditional limit theorems over an increasing set of dummy competitors and control of extreme order statistics, which we do not pursue here. Importantly, our main results for virtual dummies and the FDR guarantees of T--Rex (Sections~\ref{sec:stick_breaking} and~\ref{sec:Var_sel_FDR}) apply at arbitrary $L$ under rotational invariance and therefore do not rely on Theorem~\ref{thm:pathwise-univ}.
\end{rmrk}

\subsection{Finite Sample Effect of Gaussian Norm Fluctuations}
\label{ssec:norm-inflation}
The pathwise universality result of Theorem~\ref{thm:pathwise-univ} shows that, for a fixed number of selection steps and as $n\to\infty$, forward-selection trajectories driven by standardized i.i.d.\ dummies are asymptotically indistinguishable from those obtained with Gaussian dummies. One may therefore ask whether, in practice, it suffices to work directly with Gaussian dummies and dispense with the spherical construction and the adaptive stick-breaking machinery altogether. We show that the answer is \emph{no}: although Gaussian and spherical dummies are asymptotically equivalent, they exhibit systematically different \emph{finite-sample} behavior. The discrepancy is entirely due to random fluctuations in the norms of Gaussian dummies, which persist at moderate dimension and are amplified when many dummies compete simultaneously.

While Theorem~\ref{thm:vd-fs-dist} permits either Gaussian or spherical base laws, the two constructions differ in a subtle but important way: Gaussian dummies have random radii, whereas spherical dummies have fixed norm by construction. In practice, one might consider Gaussian columns $\bm g_\ell\sim\mathcal N(\bm 0,\tfrac{1}{m}\bm I_H)$, which have unit expected squared norm ($\E[\|\bm g_\ell\|^2]=1$) and thus match the scale of a standardized design. However, their radial component $\|\bm g_\ell\|$ still fluctuates around~$1$, introducing a \emph{norm inflation} that alters how dummies compete with real variables when $L$ is large. In contrast, spherical dummies enforce exact column standardization, whereas Gaussian dummies only match this scale in expectation.

Crucially, spherical and Gaussian dummies differ \emph{only} by this radial component. Writing $\bm g_\ell = \|\bm g_\ell\|\,\bm u_\ell$ with $\bm u_\ell\sim\Unif(S_H)$ independent of $\|\bm g_\ell\|$, we have, for any fixed $\bm r \in S_H$,
\[
M_G
:= \max_{\ell\le L}|\langle\bm g_\ell,\bm r\rangle|
\;\le\;
\bigl(\max_{\ell\le L}\|\bm g_\ell\|\bigr)\,
\underbrace{\max_{\ell\le L}|\langle\bm u_\ell,\bm r\rangle|}_{=:\,M_S}.
\]
The angular part $\bm u_\ell$ behaves exactly as a spherical dummy, while $\|\bm g_\ell\|$ scales the effective correlation strength. Since $\|\bm g_\ell\|^2 = m^{-1}\chi^2_m$, the chi-square tail bound $\Pr(\|\bm g_\ell\|^2 \ge 1 + 2\sqrt{t/m} + 2t/m) \le e^{-t}$ combined with a union bound over $\ell \le L$ at $t = \log(L/\delta)$ gives
\begin{equation}\label{eq:norm-inflation}
M_G
\;\le\;
(1+\eta_L)\,M_S,
\qquad
\eta_L=\sqrt{\tfrac{2\log(L/\delta)}{m}}+\tfrac{\log(L/\delta)}{m},
\end{equation}
with probability at least $1-\delta$. Thus, even when Gaussian dummies are scaled to unit expected norm, their radial fluctuations inflate the maximal dummy correlation by a factor $1+O(\sqrt{\log L/m})$. This bias is conservative for FDR control in the T--Rex selector (overestimation of false discoveries) but reduces power. Although $\eta_L$ decays with $m$, its magnitude remains non-negligible in realistic regimes: for $L=10^4$ and $1-\delta=0.95$, the bound predicts $7\%$--$33\%$ inflation for $m\in\{5000,300\}$. Even the lower end of this range can measurably reduce power when true signals lie near the detection boundary, as confirmed in the finite--sample simulations of Section~\ref{sec:sim}.
\section{Variable Selection and FDR Control with Virtual Dummies}
\label{sec:Var_sel_FDR}

The distributional equivalence in Theorem~\ref{thm:vd-fs-dist} ensures that any
forward-selection procedure adapted to the common filtration can be instantiated
with virtual dummies instead of an explicitly augmented dummy block without
changing the law of the resulting path. We now turn this abstract result into
concrete variable-selection and inference procedures.

We begin by presenting a generic Virtual Dummy Forward Selection (VD--FS)
template that applies to any compatible selector (Algorithm~\ref{alg:VDFS}).
We then instantiate this template for Least Angle Regression, yielding Virtual
Dummy LARS (VD--LARS) (Section~\ref{ssec:VD_LARS}). Next, we revisit the T--Rex
selector and show how its terminating random experiments, relative occurrences,
and voting thresholds can be computed entirely using VD--FS
(Section~\ref{ssec:TRex_VD}). Theorem~\ref{thm:vd-fs-dist} then implies that
all T--Rex quantities, and in particular its FDR guarantees, carry over
unchanged when explicit dummy matrices are replaced by virtual ones, yielding
Corollary~\ref{cor:fdr-rot-invariant}. Finally, we analyze the computational
complexity of VD--LARS and of T--Rex with VD--LARS, showing that the virtual
representation reduces dummy-related costs by several orders of magnitude in
large-scale settings.

\subsection{Generic Virtual Dummy Forward Selection}
\label{ssec:generic_VDFS}

Algorithm~\ref{alg:VDFS} presents the general VD--FS template.
It makes explicit the filtration transitions $\F_k \to \F_k^+ \to \F_{k+1}$ from Section~\ref{ssec:filtration} and isolates the two \emph{procedure-specific} components, the selection rule~$\phi_k$ (line~4) and the score/residual update (line~16), from the universal dummy machinery (realization, basis update, and fresh projection sampling). For linear methods, lines~9--11 use the newly selected predictor to update the basis. For GLMs and other nonlinear or robust procedures, the relevant direction is instead the current score vector, so $\bm v_k$ is updated from that score rather than from the selected predictor. Any compatible selector from Table~\ref{tab:scope_compatibility} can be instantiated in this way, while the dummy-related steps remain unchanged.

\begin{algorithm}[t]
\caption{Virtual Dummy Forward Selection (VD--FS)}
\label{alg:VDFS}
\begin{algorithmic}[1]
\Require $\bm X=(\bm x_1,\dots,\bm x_p)\subset H$, $\bm y\in H$, $L$, base law $\Psi(H)$, selection rule $\phi_k$
\State $\A_0\gets\varnothing$, $\bm r_0\gets \bm y$, $\bm e_1\gets \bm y/\|\bm y\|_2$, $\Ec_1\gets\{\bm e_1\}$, $k\gets 1$\Comment{{$\F_0$}}
\Statex\hspace{2em} \algmark{Fk0Top}
\State Draw $\{\upalpha_{11},\dots,\upalpha_{1L}\}\mid \F_0 \stackrel{\text{i.i.d.}}{\sim} \Psi_1$ and set $\tau_\ell \gets \infty$ $\forall\ell$
\Statex\hspace{2em} \algmark{Fk0Bot}
\algrbrace{Fk0Top}{Fk0Bot}{$\F_0 \to \F_{1}$}
\While{not stopped}
\State $j^\star \gets \phi_k\!\Bigl(\bm X,\bm y,\Ec_k,\bm r_{k-1},\A_{k-1},\{\upalpha_{i\ell}\}_{i\le k,\ \ell:\ \tau_\ell>k-1}\Bigr)$\Comment{{\color{red}procedure-specific}}
\State $\A_k \gets \A_{k-1}\cup\{j^\star\}$
\Statex \algmark{FkTop}
\If{$j^\star = p+\ell^\star$} \Comment{realize dummy~\eqref{eq:realize-dummy}}
  \State Draw $\bm d^{\perp}_{\ell^\star,k}\mid \F_k \sim \Psi_k^{\perp}$ and set $\tau_{\ell^\star}\gets k$
  \State $\bm d_{\ell^\star}\gets \sum_{i=1}^{k}\upalpha_{i\ell^\star}\bm e_i + \bm d^{\perp}_{\ell^\star,k}$
  \State $\bm v_k \gets \bm d_{\ell^\star}$
\Else
  \State $\bm v_k \gets \bm x_{j^\star}$
\EndIf
\State $\bm v_k^{\perp}\gets \bm v_k - \sum_{i=1}^{k}\langle \bm v_k,\bm e_i\rangle \bm e_i$\Comment{basis update}
\State $\bm e_{k+1}\gets \bm v_k^{\perp}/\|\bm v_k^{\perp}\|_2$, \quad $\Ec_{k+1}\gets \Ec_k\cup\{\bm e_{k+1}\}$
\Statex \algmark{FkBot}
\algrbrace{FkTop}{FkBot}{$\F_k \to \F_k^{+}$}
\Statex \algmark{FkPlusTop}
\State Draw $\{\upalpha_{k+1,\ell}\}_{\tau_\ell>k}\mid \F_k^{+}
\stackrel{\text{i.i.d.}}{\sim} \Psi_{k+1}$\Comment{Algorithm~\ref{alg:seq_spherical}}
\Statex \algmark{FkPlusBot}
\algrbrace{FkPlusTop}{FkPlusBot}{$\F_k^{+} \to \F_{k+1}$}

\State Update residual $\bm r_k$ \Comment{{\color{red}procedure-specific}}
\State $k\gets k+1$
\EndWhile
\end{algorithmic}
\end{algorithm}

While this paper instantiates the virtual-dummy construction for LARS (Section 5.2), the generic VD-FS template in Algorithm 2 applies to any compatible selector. Our open-source implementation at \url{https://github.com/taulantkoka/virtual-dummies} additionally provides VD-OMP, VD-AFS, and VD-AFS with logistic regression. A systematic discussion of which selectors are compatible with the framework, together with concrete examples, is given in Appendix~\ref{app:compatible} (Table~\ref{tab:scope_compatibility}).

\subsection{Virtual Dummy Least Angle Regression}
\label{ssec:VD_LARS}

The LARS algorithm is a natural candidate for the virtual dummy framework: its update rule depends only on inner products with vectors in the current span~$\V_k$, never on unrevealed dummy components in~$\V_k^\perp$. At iteration~$k$, each unrealized dummy~$\bm d_\ell$ enters only through its revealed coefficients $\upalpha_{1\ell},\dots,\upalpha_{k\ell}$ in the evolving orthonormal basis~$\Ec_k\subset H$. VD--LARS instantiates Algorithm~\ref{alg:VDFS} with the LARS selection rule and residual update. We recall the standard LARS quantities
\begin{align*}
  \bm r_k   & \text{ (residual)}                                    &
  \A_k      & \text{ (active set)}                                  &
  \bm u_k   & \text{ (equiangular direction)}                       \\
  c_{j,k}     & = \langle \bm x_j,\bm r_k\rangle \text{ (correlations)} &
  \bm a_k   & = \bm X^\top \bm u \text{ (slopes)}                   &
  \bm G_{\A_k}& = \bm X_\A^\top \bm X_\A \text{ (Gram matrix)}
\end{align*}
and $C=\max_{j\notin\A_k}\{|c_j|\}$ the current maximum correlation, $A_\A = \sqrt{\one^\top \bm G_\A^{-1}\one}$.

Let $\widetilde{\bm X}_{\!\A}$ denote the submatrix of active columns, comprising both selected real predictors and realized dummies. The equiangular direction at step $k$ is
\[
  \bm u_k = \widetilde{\bm X}_{\!\A}\,\bm w_k,
  \qquad
  \bm w_k = A_{\!\A}\,\widetilde{\bm G}_{\!\A}^{-1}\one_{\!|\A_k|},
  \qquad
  \widetilde{\bm G}_{\!\A} = \widetilde{\bm X}_{\!\A}^\top \widetilde{\bm X}_{\!\A},
\]
where $\one_{\!|\A_k|}$ is the vector of ones of length $|\A_k|$
and the sign of each active column's coefficient in $\bm w_k$
is set to match the sign of its correlation with $\bm r_k$.
By construction, $\bm u_k\in\V_k$ and makes equal angles with all active columns.
Since $\bm r_k,\bm u_k\in\V_k$, while the unrevealed residual component $\bm d_{\ell,k}^{\perp}\in\V_k^\perp$ is orthogonal to both, the correlation and slope of each unrealized dummy $\ell$ (i.e., $\tau_\ell > k$) are computable from the revealed coefficients alone:
\begin{equation}
\label{eq:vd-corr-slope}
  \rho_\ell
  = \langle \bm d_\ell,\bm r_k\rangle
  = \sum_{i=1}^k \upalpha_{i\ell}\,\langle \bm e_i,\bm r_k\rangle,
  \qquad
  a^{\mathrm{vd}}_\ell
  = \langle \bm d_\ell,\bm u_k\rangle
  = \sum_{i=1}^k \upalpha_{i\ell}\,\langle \bm e_i,\bm u_k\rangle.
\end{equation}
These are $k$-dimensional inner products (cost $\Oc(k)$ per dummy), replacing the $n$-dimensional products $\bm D^\top\bm r_k$ and $\bm D^\top\bm u_k$ of explicit augmentation. The next variable is selected as
\begin{equation}
\label{eq:lars-selection}
  j^\star
  =\operatorname*{argmin}_{\substack{j\notin\A_k\\p+\ell\notin\A_k}}
  \left\{
    \frac{C-c_j}{A_{\!\A}-a_j},\;
    \frac{C+c_j}{A_{\!\A}+a_j},\;
    \frac{C-\rho_\ell}{A_{\!\A}-a^{\mathrm{vd}}_\ell},\;
    \frac{C+\rho_\ell}{A_{\!\A}+a^{\mathrm{vd}}_\ell}
  \right\}_{\!+}, \quad C = \max_{\substack{j\notin\A_k\\p+\ell\notin\A_k}}\{|c_j|,\,|\rho_\ell|\},
\end{equation}
where $(\cdot)_+$ denotes restriction to positive values and $C$ extends the maximum correlation to include the virtual dummies. This defines $\phi_k$ in Algorithm~\ref{alg:VDFS} (line 4). Finally, the standard LARS residual update $\bm r_{k+1} = \bm r_k - \upgamma\,\bm u_k$ applies (line~16 of Algorithm~\ref{alg:VDFS}), where the step size $\upgamma = \min(\cdot)$ with the $\min(\cdot)$ taken over the same set as the $\operatorname*{argmin}(\cdot)$ in \eqref{eq:lars-selection}. Since $\bm r_k, \bm u_k\in\V_k$, the update does not interact with any unrevealed dummy component.

If the selected variable is a dummy $\ell^\star$, it is realized via~\eqref{eq:realize-dummy} and the basis is updated via~\eqref{eq:orthog} exactly as in lines~6--14 of Algorithm~\ref{alg:VDFS}; the realized dummy then joins the active set and contributes to the Gram matrix $\bm G_{\!\A}$ and subsequent equiangular computations as an ordinary predictor. Fresh projections $\{\upalpha_{k+1,\ell}\}_{\tau_\ell>k}$ are drawn from the conditional law~\eqref{eq:alpha-next} via Algorithm~\ref{alg:seq_spherical}, and the virtual correlations $\rho_\ell$ and slopes $a^{\mathrm{vd}}_\ell$ are updated via~\eqref{eq:vd-corr-slope}. The procedure terminates when (i) a prescribed number of dummies have been realized, (ii)~$|\A_k|=\min\{n,p\}$, or (iii) another LARS stopping condition is met (e.g., $C=0$).

\medskip
\noindent
Hence, VD--LARS differs from classical LARS only in that the pool of dummies is represented implicitly through sequential projections rather than explicit columns of~$\bm D$.
This implicit representation eliminates the need to form or store~$\bm D$ and yields substantial computational savings, which we quantify in Section~\ref{ssec:complexity}.
Moreover, by Theorem~\ref{thm:vd-fs-dist}, replacing explicit dummies by VD--LARS leaves the law of the entire LARS path unchanged. Consequently, any downstream FDR calibration based on dummy competition, in particular T--Rex, is preserved.

\subsection{Terminating Random Experiments with Virtual Dummies}
\label{ssec:TRex_VD}
Forward selection with dummies creates a natural competition between real
predictors and synthetic null variables, but for any fixed dummy pool this
competition is random and a single run can produce unstable selections.
The T--Rex selector stabilizes this by repeating the forward procedure on
independently generated dummy pools, each repetition called a
\emph{random experiment}.

Across $B$ such random experiments, each forward path is run until exactly
$T$ dummies have entered the model. For experiment
$b\in\{1,\dots,B\}$, let $\mathcal{C}_{b,L}(T)$ denote the set of real
variables selected before the $T$th dummy is included. The corresponding
\emph{relative occurrence} of a real variable $j\in\{1,\dots,p\}$ is
defined by
\[
  \Phi_{T,L}(j)
    := \frac{1}{B}\sum_{b=1}^B 
       \mathds{1}\!\left\{\, j \in \mathcal{C}_{b,L}(T) \,\right\},
  \qquad T\ge1,
\]
where $\mathds{1}$ is the indicator function and with the convention
$\Phi_{0,L}(j)=0$. To convert these into a selected set, the T--Rex
selector applies a \emph{voting threshold} $v\in[0.5,1)$:
\[
  \widehat{\mathcal{A}}_L(v,T)
    := \{\, j : \Phi_{T,L}(j) > v \,\}.
\]

For any \emph{fixed} $T$ and $L$, the FDR of this procedure is controlled
at level $\alpha$ as $B\to\infty$, provided the voting threshold $v$
satisfies a conservative FDP constraint (see
\citet{Machkour2025} for the precise estimator and proof).
Theorem~\ref{thm:vd-fs-dist} implies that every quantity entering this
guarantee, namely correlations, entry events, relative occurrences, and
the FDP estimator, has the same conditional distribution under virtual
dummies as under explicit augmentation. This immediately yields the
following corollary.

\begin{cor}[Virtual Dummy T--Rex]
\label{cor:fdr-rot-invariant}
Assume the setting of Theorem~\ref{thm:vd-fs-dist}. Assume further that the
assumptions of \citet[Theorem~1]{Machkour2025} hold for the corresponding
explicitly augmented T--Rex procedure. For fixed $T$ and $L$, run the T--Rex
selector with virtual-dummy forward selection $(\mathcal{P}_{\mathrm{vd}})$
based on $B$ independent random experiments, and choose the voting threshold
$v$ according to the T--Rex calibration rule. Then, as $B\to\infty$,
\[
  \mathrm{FDR}(v,T,L)
  ~=~ \E[\mathrm{FDP}(v,T,L)]
  ~\le~ \alpha.
\]
\end{cor}
\begin{proof}
The proof is deferred to Appendix~\ref{app:B}.
\end{proof}

In practice, $T$ and $v$ are chosen as
\[
  (v^\ast,T^\ast)=\operatorname*{argmax}_{v,T}\; |\widehat{\mathcal{A}}_L(v,T)|
  \quad\text{s.t.}\quad
  \widehat{\mathrm{FDP}}(v,T,L)\le\alpha,
\]
i.e., the T--Rex selector maximizes the number of selected variables subject to the conservative FDP estimator remaining below the target level. 

\subsection{Computational Complexity of VD--LARS}
\label{ssec:complexity}

In standard LARS one would conceptually extend the design to $(\bm X\ \bm D)$, where $\bm D\in\R^{n\times L}$ contains explicit dummy columns. In VD--LARS, this block is replaced by a compact projection matrix $\bm A_k\in\R^{k\times L}$, whose $\ell$th column $(\upalpha_{1\ell},\dots,\upalpha_{k\ell})^\top$ stores the coefficients of the $\ell$th virtual dummy in the current orthonormal basis $\Ec_k=\{\bm e_1,\dots,\bm e_k\}\subset H$. Writing $\bm E_k=(\bm e_1\,\dots\,\bm e_k)\in\R^{n\times k}$, the dummy representation is low rank if $k\ll n$\footnote{We use $\Oc(\cdot)$ notation for asymptotic cost to avoid unnecessary distinctions between upper bounds and exact asymptotic orders, even though the underlying linear-algebra kernels admit tight $\Theta(\cdot)$ bounds.}.

\paragraph{Memory complexity.}
The augmented formulation requires storing or regenerating an $n\times L$ dummy block, in addition to the real design:
\[
\mathcal{M}_{\text{aug}} = \Oc(np) + \Oc(nL).
\]
In VD--LARS, the dummy block $\bm D$ is never stored. Instead, at step $k$ we retain the orthonormal basis $\bm E_k\in\R^{n\times k}$, the dummy projection
matrix $\bm A_k\in\R^{k\times L}$, and the $T$ realized dummy columns:
\[
\mathcal{M}_{\text{vd}}
  = \Oc(np)
  + \Oc(nk)
  + \Oc(kL)
  + \Oc(nT),
\]
where $T\le k$ is the number of realized dummies. For $k\ll n$ and $T\ll L$, the dummy-related storage drops from $\Oc(nL)$ to $\Oc(kL+nT)$, i.e.\ by a factor $n/k$ (up to the additional $\Oc(nk)$ cost of representing the revealed subspace).

\paragraph{Per--step time complexity.}
Each LARS step involves computing correlations, updating the Cholesky factor, and advancing along the equiangular direction. For the real features, both algorithms require   (i) one Gram--Schmidt update and two triangular solves: $\Oc(k^2)$ flops, (ii) one dense matrix--vector product $\,\bm a=\bm X^\top\bm u\,$: $\Oc(np)$ flops (dominant term).

The difference lies in the dummy updates: for augmented LARS computing $\bm a_{\mathrm{dum}} = \bm D^\top \bm u$ is $\Oc(nL),$ while in VD--LARS we compute $\bm a_{\mathrm{vd}} = \bm A_k^\top (\bm E_k^\top \bm u)$, which is only $\Oc(kL)$. If a dummy is realized at step $k$, completing it requires projecting a sampled $n$-vector onto $\V_k^\perp$ via $\bm I - \bm E_k\bm E_k^\top$, which costs $\Oc(nk)$. This is incurred only $T$ times along the path. Thus the per--step complexities are:
\[
\begin{aligned}
&\text{Augmented LARS:}\quad
&\Oc(np + k^2 + nL)&, \\[0.5ex]
&\text{VD--LARS:}\quad
&\Oc(np + nk + k^2 + kL)&
\quad\text{(plus $\Oc(nk)$ on $T$ realization steps)}.
\end{aligned}
\]
The $\Oc(nk)$ term in VD--LARS is the cost of maintaining an explicit representation of the revealed span $\V_k$ in ambient coordinates, which is not required in textbook LARS but is negligible compared to $nL$ whenever $k\ll L$.

\paragraph{Example at biobank scale.}
For representative large-cohort settings ($n=5{\times}10^5$, $p=10^6$, $L\ge p$, $k=50$), explicit augmentation requires $nL=5{\times}10^{11}$ dummy inner products per step and storing an $n\times L$ dummy matrix, i.e., about $4$~TB in \texttt{float64}. In contrast, VD--LARS reduces dummy work to $kL=5{\times}10^{7}$ inner products per step and stores only the dummy projection state $\bm A_k\in\mathbb R^{k\times L}$, i.e., about $400$~MB in \texttt{float64}, an approximately $n/k\approx 10^4$ reduction. The real--feature correlation $\bm a=\bm X^\top\bm u$ still costs $\Oc(np)$ per step ($\approx 5{\times}10^{11}$ flops) and dominates. VD--LARS leaves this term unchanged but removes the $\Oc(nL)$ dummy burden, replacing it with $\Oc(kL)$ with $k\ll n$.

\subsection{Complexity of the T--Rex Selector with VD--LARS}

The T--Rex selector runs $B$ early-terminated forward experiments. Each experiment includes the real predictors and $L$ dummies, and is stopped as soon as $T$ dummies have entered the active set. Hence its computational cost depends on the (random) forward path length $\kappa$ and on the number $T$ of selected dummies. To characterize the expected path length, we rely on the negative hypergeometric (NHG) model used in the original T--Rex analysis.  Let $p_1$ and $p_0=p-p_1$ be the numbers of true active and true null predictors.  Then the number of null predictors examined before $T$ dummies have been selected satisfies\footnote{We write $\Xi \sim \mathrm{NHG}(N, K, T)$ for the number of successes observed before the $T$-th failure when sampling without replacement from a population with $K$ successes and $N-K$ failures.}
\begin{align*}
\mathbb E[\kappa]
  \;\le\; p_1 + T + \mathbb E[\Xi]
  \;=\; &p_1 + T + \frac{T}{L+1}\,p_0
  \;\le\; p_1 + 2T,
\end{align*}
for $\Xi \sim \mathrm{NHG}(p_0+L,\,p_0,\,T).$ Under the standard sparsity regime $p_1,T=\Oc(1)$ and $L=\eta p$ with fixed $\eta\geq1$, the expected path length $\mathbb E[\kappa]$ is $\Oc(1)$.

\paragraph{Expected time complexity.}
For a path of length $\kappa$, VD--LARS performs   (i) $\Oc(np\kappa)$ work on real features, (ii) $\Oc(L\kappa^2)$ work on virtual dummy correlations, (iii) $\Oc(n\kappa T)$ work for the $T$ dummy realizations. Thus, we have
\[
\mathbb E[\mathcal C_{\mathrm{T\text{-}Rex}}]
 = \Oc\!\Bigl(
   np\,\mathbb E[\kappa]
   + L\,\mathbb E[\kappa^2]
   + nT\,\mathbb E[\kappa]
 \Bigr)\ = \Oc(np + L),
\]
since $\mathbb E[\kappa]=\Oc(1)$ and $T=\Oc(1)$. Because $L=\eta p$ in typical calibrations, the dominant term is $\Oc(np)$, matching the stated complexity of the original T--Rex selector while reducing dummy-related work by $n/k$.

\paragraph{Expected peak memory.}
At step $k$, VD--LARS stores the design $\bm X$ ($\Oc(np)$), the projection rows $\bm A_k$ ($\Oc(kL)$), and the realized dummies ($\Oc(nT_k)$). Hence the peak memory over a path of length $\kappa$ is
\[
\mathcal M_{\mathrm{vd}}^{\max}
 = \Oc(np + L \kappa + nT) = \Oc(np + L),
\]
using the NHG bound $\mathbb E[K]=\Oc(1)$. Compared to augmented LARS, which requires $\Oc(np + nL)$, this replaces the dummy cost $nL$ by only $L$, i.e., a reduction by a factor of $n$.

\section{Simulations}
\label{sec:sim}
This section presents a simulation study validating both the theoretical and practical contributions of the virtual dummy construction. First, we empirically confirm the key finite--sample equivalences: (i) the pathwise distributional equivalence between virtual and explicitly augmented dummy forward selection (Theorem~\ref{thm:vd-fs-dist}), and (ii) the resulting equivalence of the full T--Rex selector, including stopping and FDR calibration (Corollary~\ref{cor:fdr-rot-invariant}). We then probe the conditional CLT mechanism underlying Theorem~\ref{thm:pathwise-univ} by measuring the Gaussianity of fresh dummy projections under non-Gaussian coordinate laws. Next, we quantify the finite--sample discrepancy between Gaussian and spherical dummy base laws discussed in Section~\ref{ssec:norm-inflation}. Finally, we benchmark memory usage and runtime of VD--LARS against explicit augmentation.

\subsection{Empirical Validation of Distributional Equivalence}
\label{ssec:dist_eq_sim}
We begin by empirically validating the finite--sample distributional equivalence of VD--LARS and its explicitly augmented counterpart (AD--LARS), as established by Theorem~\ref{thm:vd-fs-dist}.  The goal of this experiment is not to assess statistical power, but to verify that the two implementations generate indistinguishable forward--selection trajectories when run under the same model and stopping rule. In each Monte Carlo replicate, we generate a linear model with $n=300$ observations and $p=1000$ standardized real predictors.  The true active set has size $|\mathcal A|=10$, and the signal-to-noise ratio is set to $\mathrm{SNR}=1$.  We fix $(\bm X,\bm y)$ and resample only the dummy ensemble across replicates.  Each run uses $L=5p$ dummy variables, and we repeat the experiment over $n_{\mathrm{MC}}=2000$ independent dummy realizations.

To probe agreement along the selection path, we stop the algorithm after exactly $T$ dummies have entered the active set and examine the remaining unselected dummies.  Specifically, we compute order statistics of the absolute dummy--residual correlations, \( |\rho_\ell (T)| = |\langle \bm d_\ell,\bm r_T\rangle|,\) namely the ranks $1,5,20$, and $50$ among the unselected dummies.  We further compare the full marginal distributions at two representative stopping times, $T=1$ and $T=20$, using empirical cumulative distribution functions (ECDFs) and quantile--quantile (Q--Q) plots.

Figure~\ref{fig:distr_equivalence} shows nearly perfect agreement between VD--LARS and AD--LARS across all diagnostics.  In particular, both the trajectory of the order statistics as a function of $T$ and the marginal distributions of dummy correlations at fixed $T$ coincide within Monte Carlo error. Since the dummy selection times are random and depend on the entire history of the forward path, this experiment provides strong empirical evidence for the pathwise distributional equivalence of VD--LARS and explicit dummy augmentation.

\begin{figure}[t]
    \centering
    \includegraphics[width=
    \linewidth]{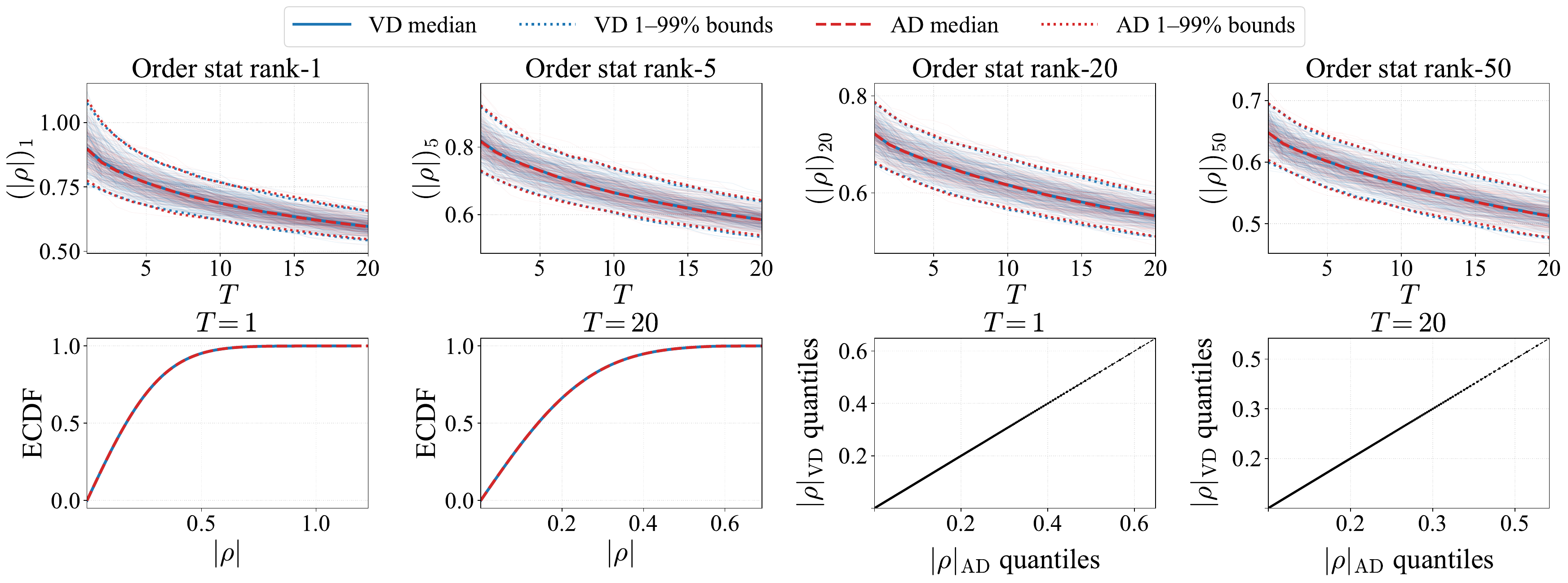}
    \caption{ Distributional equivalence of VD--LARS and explicitly augmented AD--LARS. Top row: median and $1-99\%$ Monte Carlo envelopes (over $2000$ replicates) of order statistics of the absolute correlations between the residual and the unselected dummies after $T$ dummies have been selected, shown for ranks $1, 5, 20,$ and $50$. Bottom row: empirical CDFs and Q--Q plots comparing the marginal distribution of $|\langle \bm d_\ell,\bm r_T\rangle|$ at $T=1$ and $T=20$. Across all diagnostics, VD--LARS and AD--LARS are indistinguishable within Monte Carlo error, supporting the pathwise distributional equivalence predicted by Theorem~\ref{thm:vd-fs-dist}.
    }
    \label{fig:distr_equivalence}
\end{figure}

\subsection{Empirical Validation of FDR Control for Virtual Dummy T--Rex}
\label{ssec:fdr_sim}

We next validate Corollary~\ref{cor:fdr-rot-invariant} by comparing the full
T--Rex selector implemented with explicitly augmented dummies (AD--T--Rex) to
its virtual dummy counterpart (VD--T--Rex).  While the previous experiment
tests distributional equivalence at the level of a single LARS path, this
experiment probes equivalence at the level of the complete T--Rex procedure,
including its calibration, stopping rule, and FDR control.

We consider the same linear model as in
Section~\ref{ssec:dist_eq_sim}, with $n=300$, $p=1000$, and a sparse active set
of size $s=10$.  The design is random with standardized columns, and the
SNR is set to values between $0.2$ and $5$.  For each SNR value, we run the T--Rex selector with $B=20$ independent random experiments at target FDR levels $\alpha\in\{0.1,0.05,0.01\}$ and the total number of dummies takes values in $L\in\{p,5p,10p,20p,30p,40p\}$. Results are averaged over $n_{\mathrm{MC}}=1000$ Monte Carlo replicates for $\alpha\in\{0.1,0.05\}$ and $n_{\mathrm{MC}}=100$ replicates for $\alpha=0.01$. For each replicate, we record the FDP and TPP of the final T--Rex selected set. Figure~\ref{fig:fdr_equivalence} shows that VD--T--Rex and AD--T--Rex are practically indistinguishable across all SNR values and dummy sizes, and that both implementations control the FDP well below the target levels. The sweep over $L$ also illustrates an important practical point: increasing the number of dummies can substantially improve power, especially at stringent target levels.  For instance, at $\alpha=0.05$, using only $L=p$ yields essentially no power, whereas larger dummy budgets recover a nontrivial TPP curve.  This effect is even more pronounced at $\alpha=0.01$, where achieving power requires very large dummy pools.  These regimes are precisely where explicit augmentation becomes computationally prohibitive, and therefore where the virtual dummy construction removes the primary bottleneck of T--Rex while preserving its selection law and FDR guarantees.
\begin{figure}[t]
    \centering
    \includegraphics[width=
    \linewidth]{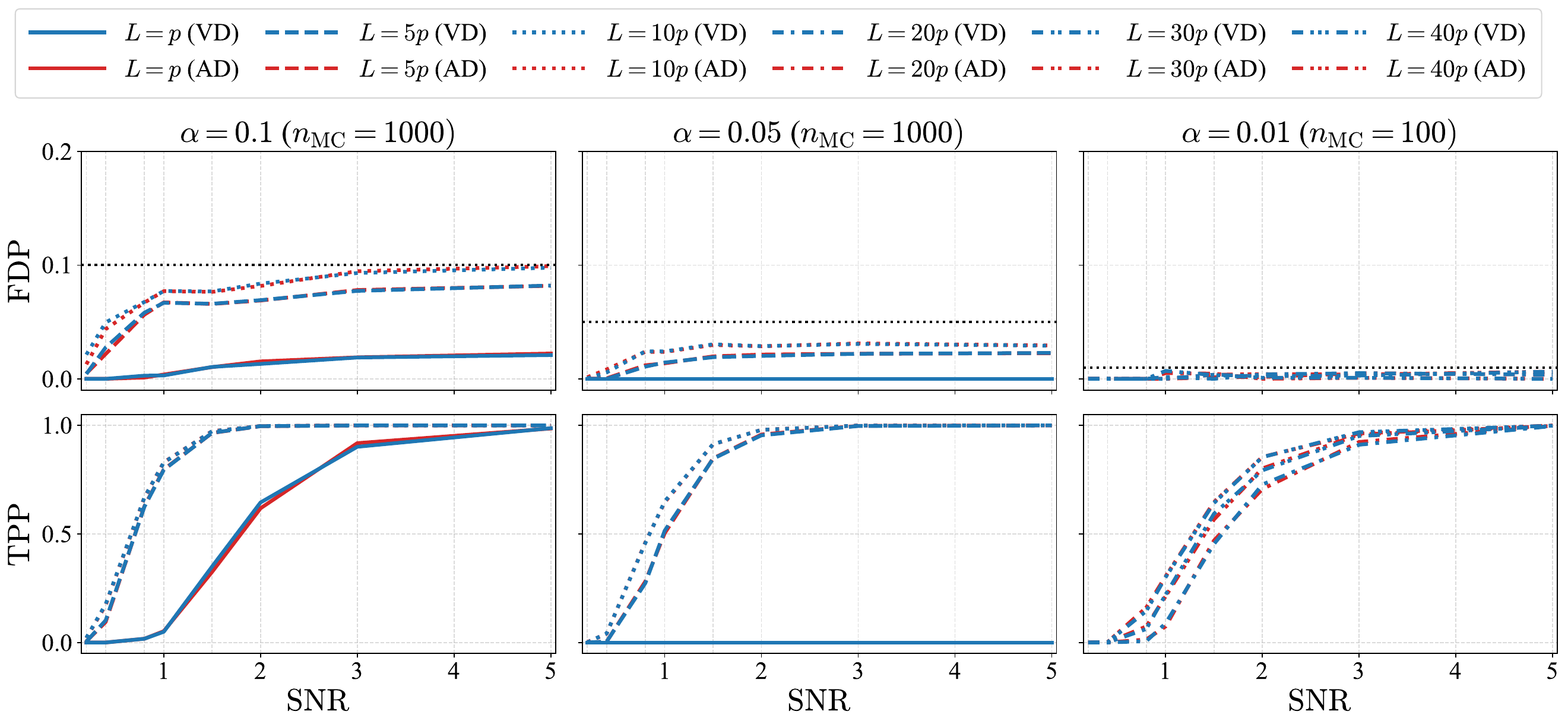}
    \caption{
    FDR control and power of VD--T--Rex versus AD--T--Rex.
    Top row: empirical averaged false discovery proportion (FDP).
    Bottom row: empirical averaged true positive proportion (TPP).
    Columns correspond to target levels $\alpha\in\{0.1,0.05,0.01\}$, and curves
    show varying number of dummies  $L\in\{p,5p,10p,20p,30p,40p\}$.  The horizontal dotted line indicates the target FDR level $\alpha$.  Across all settings, VD--T--Rex and AD--T--Rex produce indistinguishable FDP and TPP curves, confirming that the virtual dummy implementation preserves the full T--Rex selection law and its FDR
    calibration.
    }
    \label{fig:fdr_equivalence}
\end{figure}
\begin{figure}[t]
    \centering
    \includegraphics[width=\linewidth]{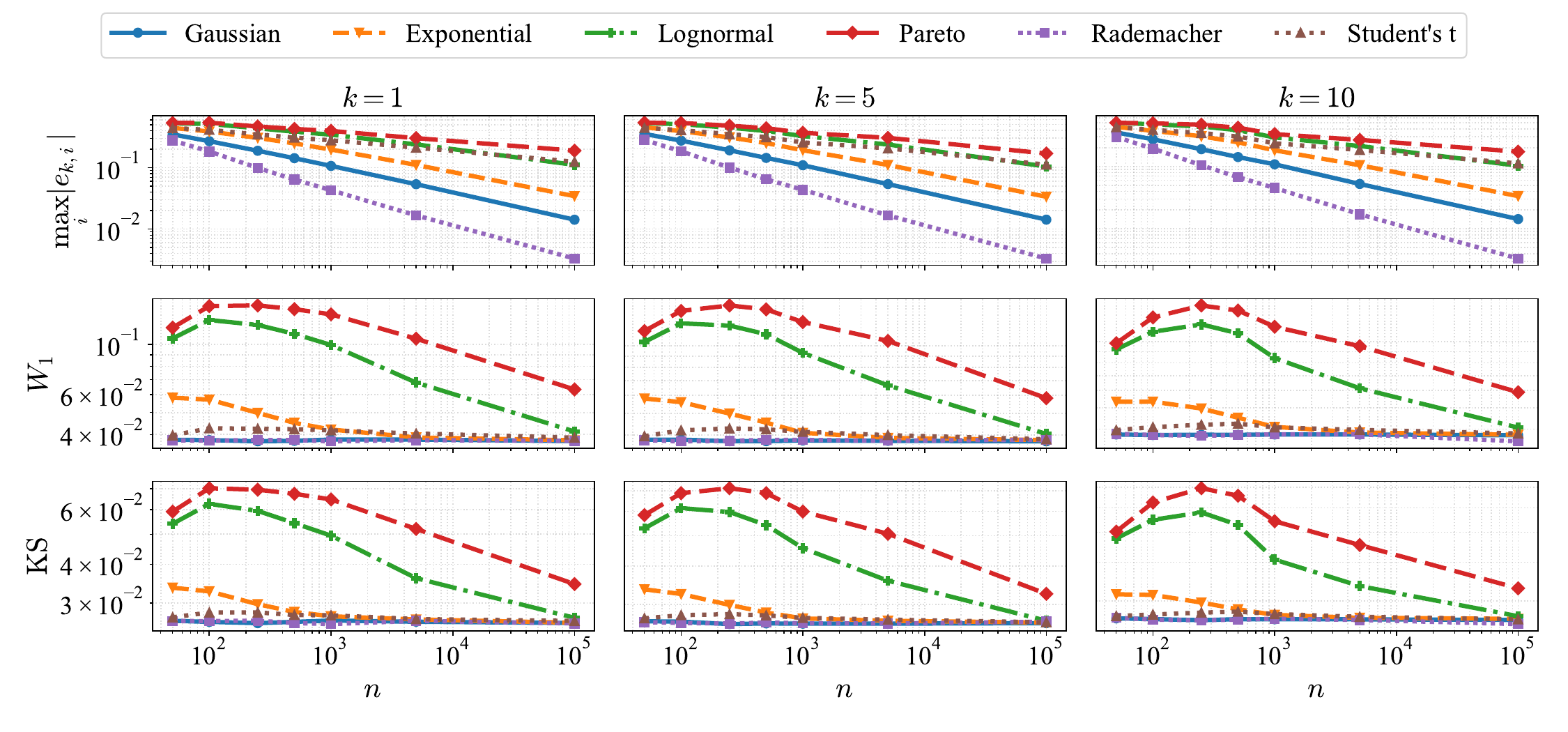}
    \caption{
    Universality diagnostics for Lemma~\ref{lem:fresh-step} (conditional CLT for fresh projections). Columns correspond to forward--selection steps $k\in\{1,5,10\}$, and curves compare different standardized i.i.d.\ coordinate laws used to generate the dummy variables (standard Gaussian, Rademacher, exponential, Student's $t_3$, lognormal, Pareto).  Top row: delocalization proxy $\max_i |(\bm e_k)_i|$ for the data--dependent direction $\bm e_k$.  Middle row: median 1--Wasserstein distance between the empirical distribution of fresh dummy projections $\langle \bm d_\ell,\bm e_k\rangle$ (over unselected dummies) and $\mathcal N(0,1)$.  Bottom row: corresponding Kolmogorov--Smirnov statistic. Across all distributions, $\bm e_k$ becomes delocalized as $n$ grows and the projection distribution approaches $\mathcal N(0,1)$, with slower convergence for heavier-tailed laws at moderate $n$.}
    \label{fig:univ_lemma6}
\end{figure}

\subsection{Conditional Gaussianity of Dummy Projections}
\label{ssec:univ_diag}

We now provide empirical evidence for the conditional CLT mechanism underlying the pathwise universality statement of Theorem~\ref{thm:pathwise-univ}.  The key intermediate claim (Lemma~\ref{lem:fresh-step} in our proof outline) is that, at each step $k$, the fresh projections of the unselected standardized i.i.d.\ dummies onto the data--dependent direction $\bm e_k$ are approximately Gaussian when $n$ is large, provided that the direction $\bm e_k$ is sufficiently delocalized.

We generate a random Gaussian design with $p=100$ standardized real variables and a centered Gaussian response, run forward selection for $K=10$ steps, and augment with $L=10p$ standardized i.i.d.\ dummies.  We sweep the sample size over $n\in\{50,100,250,500,1000,5000,10^5\}$ and repeat each configuration over $n_{\mathrm{MC}}=5000$ independent dummy ensembles.  To probe the conditional CLT, we consider six coordinate laws for the raw dummies: standard Gaussian, Rademacher, exponential, Student's $t$ ($\nu_{t}=3$), lognormal ($\sigma_{\mathrm{lognorm}}=1$), and Pareto ($\alpha_{\mathrm{Pareto}}=3$), each centered and scaled to unit sample variance before column standardization.

\begin{figure}
    \centering
    \includegraphics[width=0.9\linewidth]{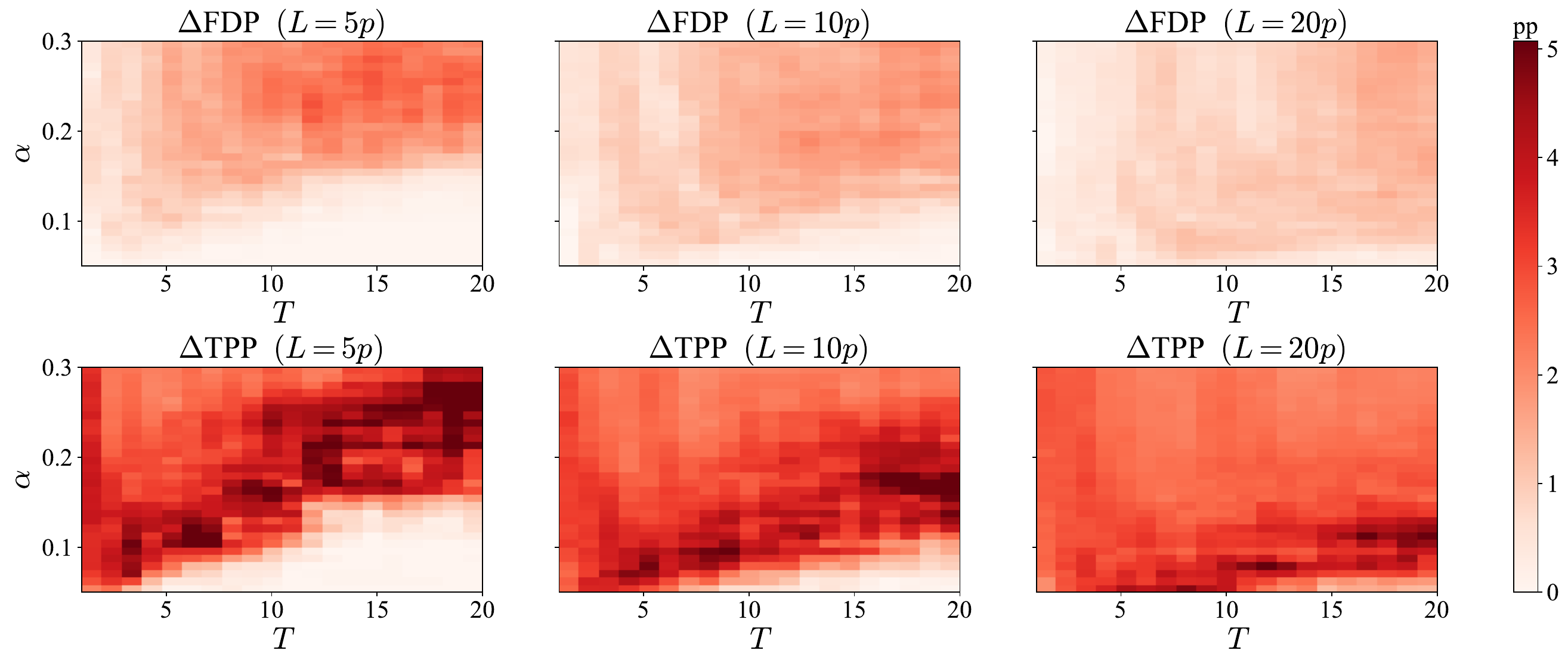} 
\caption{
Finite--sample effect of Gaussian norm fluctuations in T--Rex. Heatmaps show differences, reported as percentage points (pp), between Gaussian and spherical dummy base laws in a high-dimensional regime ($n=300$, $p=10^4$, $s=10$, $\mathrm{SNR}=1$), averaged over $M=1000$ replicates. Top row: $\Delta\mathrm{FDP}=\mathrm{FDP}_{\mathrm{Sph}}-\mathrm{FDP}_{\mathrm{Gauss}}$. Bottom row: $\Delta\mathrm{TPP}=\mathrm{TPP}_{\mathrm{Sph}}-\mathrm{TPP}_{\mathrm{Gauss}} $. Columns correspond to number of dummies $L\in\{5p,10p,20p\}$, and axes sweep the T--Rex calibration parameters $(\alpha,T)$. Gaussian dummies are systematically more conservative (lower FDP) but lose power relative to spherical dummies.
}
\label{fig:gauss_stick_tpp}
\end{figure}

For steps $k\in\{1,5,10\}$, we record (i) the delocalization proxy $\max_i |(\bm e_k)_i|$, and (ii) distributional distances between the empirical distribution of fresh dummy projections $\{\langle \bm d_\ell,\bm e_k\rangle:\ell\ \text{unselected at step }k\}$ and the standard normal law, using the Kolmogorov--Smirnov statistic and the 1--Wasserstein distance.  Figure~\ref{fig:univ_lemma6} shows that $\max_i |(\bm e_k)_i|$ decreases with $n$, consistent with Assumption~(D), and that both KS and Wasserstein distances shrink toward~$0$ across all dummy distributions.  As expected, heavier-tailed coordinate laws (Pareto, lognormal, $t_3$) converge more slowly at moderate $n$, but the same Gaussian limit emerges as $n$ grows.  These results support the conditional Gaussian approximation for fresh projections that drives the universality argument.

\subsection{Gaussian vs.\ Spherical Dummies}
\label{ssec:norm_inflation_sim}

The equivalence result of Theorem~\ref{thm:vd-fs-dist} applies to both Gaussian and spherical dummy base laws.  Since Gaussian dummies are easier to generate, one might be tempted to dispense with the spherical stick-breaking construction altogether. However, as discussed in Section~\ref{ssec:norm-inflation}, we expect a systematic finite--sample discrepancy: Gaussian dummies have random radii, and when many dummies compete simultaneously, the largest dummy norms inflate the maximal dummy correlations.  This makes Gaussian dummies slightly more competitive than spherical ones, leading to more conservative selections and a loss of power.

To isolate this effect, we run the T--Rex selector with either spherical dummies (generated by adaptive stick-breaking) or Gaussian dummies scaled so that $\E\left[\|\bm g_\ell\|_2^2\right] = 1$. We consider a high-dimensional regime with $n=300$, $p=10^4$, $s=10$, and $\mathrm{SNR}=1$, and average results over $n_{\mathrm{MC}}=1000$ Monte Carlo replicates. In each replicate, T--Rex is run with $B=20$ independent random experiments. We sweep the target level $\alpha$ and the number of selected dummies $T$ over a grid, and repeat the experiment for $L\in\{5p,10p,20p\}$ dummies.

Figure~\ref{fig:gauss_stick_tpp} reports the differences between Gaussian and spherical dummies, $\Delta\mathrm{FDP}$ and $\Delta\mathrm{TPP}$, where positive values indicate that Gaussian dummies yield higher FDP or higher power. Across all settings, Gaussian dummies are consistently more conservative: the FDP is slightly reduced relative to spherical dummies, but this comes at a measurable cost in power, with losses reaching several percentage points in the most stringent calibration regimes. These results confirm that, although Gaussian and spherical dummies are asymptotically equivalent, the spherical construction yields a tangible finite--sample advantage when $p$ is large.

\subsection{Memory and Runtime Scalability Benchmark}
\label{ssec:bench_runtime_memory}

We conduct a systems benchmark comparing the computational footprint of AD--LARS and VD--LARS.  The purpose of this experiment is to quantify the practical impact of eliminating the explicit $n\times L$ dummy block, which is the dominant bottleneck in dummy-augmented forward selection at large sample sizes.

We implemented both AD--LARS and VD--LARS in \texttt{C++} and benchmarked them on a Mac Studio (M1 Ultra, 20 CPU cores, 128\,GB RAM).  Each run was executed in a separate single-threaded process (no parallelism). To reduce run-to-run variability, processes were executed with high scheduling priority (favoring performance cores), and no other user processes were active. Peak resident set size (RSS) was capped at 100\,GB.  We fix a linear model with $s=10$ active variables and $\mathrm{SNR}=1$, and run the procedure until $T=10$ dummies have entered the active set.  We sweep the number of dummies $L\in\{p,5p,10p\}$ and report medians over 100 replicates per configuration.  For AD--LARS, runtime measurements exclude the cost of generating and standardizing the dummy matrix (since dummies could in principle be streamed). Figure~\ref{fig:bench_vdlars_adlars} reports (top) peak RSS, together with the theoretical uncompressed memory required to store $(\bm X,\bm y)$ (dotted baseline) and the 100\,GB cap; (middle) the corresponding memory overhead above this baseline; and (bottom) wall-clock runtime.  Across all settings, VD--LARS tracks the $(\bm X,\bm y)$ baseline closely: the dummy-related memory overhead is reduced from tens of gigabytes under explicit augmentation to a few hundred megabytes, consistent with storing only the $k\times L$ projection coefficients and the small number of realized dummies.  In contrast, AD--LARS exhibits pronounced growth in both memory and runtime as $L$ increases, reflecting the unavoidable cost of explicitly materializing, storing, and repeatedly correlating extremely large dummy blocks.  In concrete terms, VD--LARS reduces both peak memory overhead and wall-clock runtime by one to two orders of magnitude relative to AD--LARS across the configurations tested.  Notably, the dependence on $L$ becomes negligible for VD--LARS when $n$ is large, whereas it remains substantial for AD--LARS.

These results corroborate the main computational claim of the paper: virtual dummies remove the primary scaling barrier of dummy-augmented forward selection.  This is particularly consequential in the stringent-FDR regimes of Figure~\ref{fig:fdr_equivalence}, where achieving power requires very large dummy pools, and in biobank-scale settings where $n$ may be in the hundreds of thousands while only a few hundred to a few thousand variables are ever selected.

\begin{figure}[t]
    \centering
   \includegraphics[width = \linewidth]{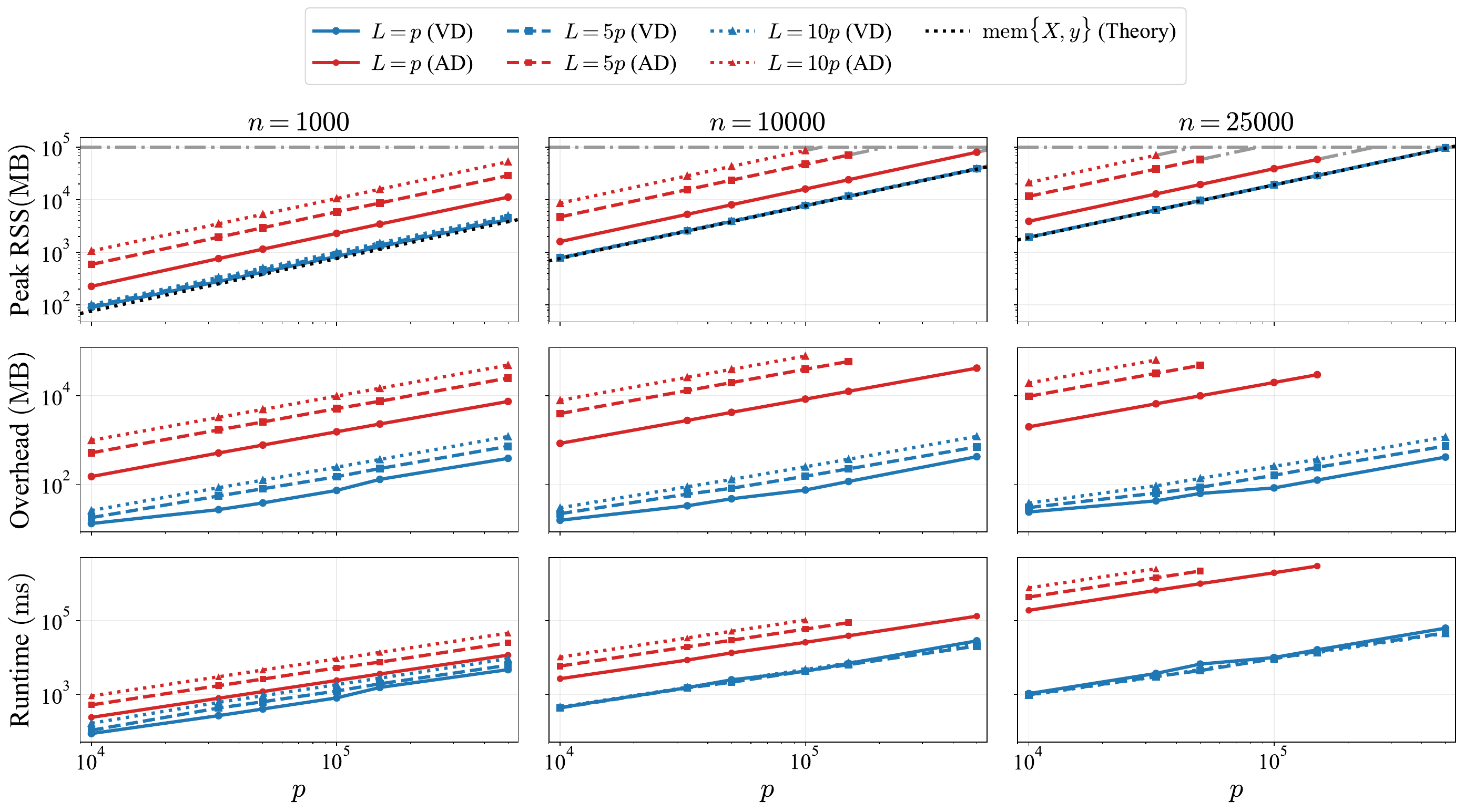}
    \caption{Benchmark of AD--LARS and VD--LARS. Columns correspond to $n\in\{10^3,10^4,2.5\times 10^4\}$, and the horizontal axis sweeps the number of real predictors $p$. Curves show number of dummies $L\in\{p, 5p, 10p\}$. Top row: peak resident set size (RSS), with the theoretical uncompressed memory required to store $(\bm X,\bm y)$ shown as a dotted baseline and the 100\,GB RSS cap indicated by a horizontal line. Middle row: memory overhead above the $(\bm X,\bm y)$ baseline. Bottom row: wall-clock runtime. For AD--LARS, the reported runtime excludes dummy generation and standardization. VD--LARS closely tracks the $(\bm X,\bm y)$ baseline and exhibits negligible dependence on $L$ at large $n$, whereas AD--LARS incurs substantial $L$-dependent memory overhead and pronounced runtime growth.}
    \label{fig:bench_vdlars_adlars}
\end{figure}

\subsection{Benchmark on Realistic GWAS Data}\label{sec:hapnest_benchmark}

\begin{table}[h]
\centering
\caption{FDR benchmark on HAPNEST-simulated GWAS data with multiplicative relative risk model (heterozygous relative risk $\in [1.05, 1.25]$, $s = 10$ causal SNPs, target FDR level $\alpha = 0.1$). Results are averaged over 100 (small-scale) and 30 (full-scale) independent Monte Carlo replications. Methods marked with~$\dagger$ did not complete a single run within the 24\,h wall-time limit; the number of completed runs is given in parentheses. Memory is reported as peak resident set size (RSS).}
\label{tab:hapnest_benchmark}
\vspace{0.5em}
\small
\setlength{\tabcolsep}{4pt}
\renewcommand{\arraystretch}{1.15}
% ── Small-scale panel ──
\textbf{(a) Small-scale:} $n = 10\,000$, $p \approx 31\,000$, $50\%$ prevalence, $100$~replicates
\vspace{0.3em}
\begin{tabular*}{\textwidth}{@{}l @{\extracolsep{\fill}} c c r r r r@{}}
\toprule
\textbf{Method} & $\overline{\textbf{FDP}} \le 10\%$ & $\overline{\textbf{TPP}} > 0\%$ & $\overline{\textbf{FDP}}$ & $\overline{\textbf{TPP}}$ & \textbf{Avg.\ time} & \textbf{Peak RSS} \\
 & & & (in \%) & (in \%) & (hh:mm:ss) & (in GB) \\
\midrule
\textbf{VD--T--Rex}
  & \textbf{\cmark} & \textbf{\cmark} & \textbf{5.2} & \textbf{50.6} & \textbf{00:02:29} & \textbf{5.6} \\
BH--CA
  & \cmark & \xmark & 0.0 & 0.0 & 00:00:03 & 7.2 \\
BY--CA
  & \cmark & \xmark & 0.0 & 0.0 & 00:00:03 & 7.2 \\
BH--HD
  & \xmark & \cmark & 11.0 & 35.5 & 00:02:06 & 8.6 \\
BY--HD
  & \cmark & \cmark & 1.0 & 28.4 & 00:02:06 & 8.6 \\
Knockoffs--HD
  & \cmark & \xmark & 0.0 & 0.0 & 00:04:34 & 9.3 \\
Model-X Knockoffs
  & \cmark & \xmark & 0.0 & 0.0 & 23:16:30 & 95.9 \\[-0.3em]
  & & & & & \multicolumn{1}{r}{\footnotesize (5/100 completed)} \\
\bottomrule
\end{tabular*}
\vspace{1.2em}
% ── Full-scale panel ──
\textbf{(b) Full-scale:} $n = 100\,000$, $p \approx 394\,000$, $10\%$ prevalence, $30$~replicates
\vspace{0.3em}
\begin{tabular*}{\textwidth}{@{}l @{\extracolsep{\fill}} c c r r r r@{}}
\toprule
\textbf{Method} & $\overline{\textbf{FDP}} \le 10\%$ & $\overline{\textbf{TPP}} > 0\%$ & $\overline{\textbf{FDP}}$ & $\overline{\textbf{TPP}}$ & \textbf{Avg.\ time} & \textbf{Peak RSS} \\
 & & & (in \%) & (in \%) & (hh:mm:ss) & (in GB) \\
\midrule
\textbf{VD--T--Rex} \texttt{(memmap)}
  & \textbf{\cmark} & \textbf{\cmark} & \textbf{5.8} & \textbf{59.1} & \textbf{17:49:03} & \textbf{16.4} \\[-0.3em]
  & & & & & \multicolumn{1}{r}{\footnotesize (23/30 completed)} \\
BH--CA \texttt{(memmap)}
  & \cmark & \xmark & 0.0 & 0.0 & 00:41:12 & 4.1 \\
BY--CA \texttt{(memmap)}
  & \cmark & \xmark & 0.0 & 0.0 & 01:09:17 & 4.3 \\
BH--HD$^\dagger$
  & --- & --- & --- & --- & \multicolumn{1}{r}{\footnotesize 0/30 (timed out)} & --- \\
BY--HD$^\dagger$
  & --- & --- & --- & --- & \multicolumn{1}{r}{\footnotesize 0/30 (timed out)} & --- \\
Knockoffs--HD$^\dagger$
  & --- & --- & --- & --- & \multicolumn{1}{r}{\footnotesize 0/30 (timed out)} & --- \\
Model-X Knockoffs$^\dagger$
  & --- & --- & --- & --- & \multicolumn{1}{r}{\footnotesize 0/30 (timed out)} & --- \\
\bottomrule
\end{tabular*}
\end{table}

We now evaluate VD--T--Rex on genotype data with realistic linkage disequilibrium (LD) structure, simulated using \textsc{HAPNEST}~\citep{Hapnest2023}. \textsc{HAPNEST} accurately reproduces the LD patterns, minor allele frequency (MAF) spectrum, and haplotype diversity of real human populations. 
Genotypes were simulated with European-ancestry haplotype reference panels and subjected to standard per-chromosome quality control (MAF ${>}1\%$, call rate ${>}95\%$, Hardy--Weinberg $p-\text{value} > 10^{-6}$), followed by LD pruning via single-linkage clustering at absolute pairwise correlation $|\hat{c}| > 0.7$ to remove near-collinear SNPs while preserving realistic residual LD.\footnote{Pairwise correlations were estimated on a re-standardized subsample of $3\,000$ rows; one representative per cluster was retained at random. Full preprocessing details  and code are available at \url{https://github.com/taulantkoka/virtual-dummies}.} Design matrices were column-centered and $\ell_2$-normalized.

We consider two settings: a small-scale setting with $n = 10\,000$ individuals and $p \approx 31\,000$ SNPs after preprocessing (chromosome~1 only, 100 replicates), and a full-scale setting with $n = 100\,000$ and $p \approx 394\,000$ SNPs after preprocessing (all $22$ autosomes, $30$ replicates). Binary case--control phenotypes are generated from $s = 10$ weak causal SNPs under a multiplicative relative risk model ($\mathrm{RR}_j \sim \mathrm{Unif}(1.05, 1.25)$), with prevalence $0.5$ and $0.1$ at small and full scale, respectively. All experiments ran on exclusive nodes of the Lichtenberg HPC cluster at TU Darmstadt ($384$\,GB RAM, $96$ cores) with a $24$\,h wall-time limit per replicate.

We compare VD--T--Rex against six competing approaches, all evaluated at target FDR level $\alpha = 0.1$. The first two are marginal testing procedures: \emph{BH--CA} and \emph{BY--CA} apply the Cochran--Armitage trend test independently to each SNP and correct the resulting $p$-values with the Benjamini--Hochberg~\citep{Benjamini1995} or Benjamini--Yekutieli~\citep{Benjamini2001} procedure, respectively; the BY correction additionally guarantees FDR control under arbitrary dependence among test statistics. The next two methods, BH--HD and BY--HD, follow the sample-splitting framework of~\cite{Barber2019} for high-dimensional inference. The sample is split using stratified sampling ($37.5\,\%$ screening, $62.5\,\%$ inference), preserving the case--control ratio in both halves. On the screening split, a Lasso path via LARS selects the first $k_{\max} = \lfloor 0.25 \cdot n_{\mathrm{infer}} \rfloor$ variables to enter the regularization path, and on the inference split, OLS is fit on the screened SNPs and two-sided $t$-tests are corrected with BH or BY at level~$\alpha$, respectively. Knockoffs--HD uses the same screening but replaces OLS inference with fixed-$X$ knockoffs~\citep{Barber2015} on the held-out data, following the high-dimensional knockoff approach of~\cite{Barber2019}. Finally, Model-X Knockoffs uses the full \texttt{knockpy} implementation~\citep{Spector2022} without data splitting: it estimates and inverts the $p \times p$ covariance matrix to construct second-order Gaussian knockoff copies and applies the knockoff filter with Lasso coefficient-difference statistics. The VD--T--Rex selector uses VD--LARS in $B = 20$ independent random experiments with $L = 5p$ dummies per experiment. At full scale, the standardized genotype matrix is stored as a Fortran-order \texttt{float64} memory-mapped file. Both VD--T--Rex and the marginal baselines (BH--CA, BY--CA) access this file column-by-column, avoiding materialization of the full $294\,$GB matrix in RAM. 
 
Table~\ref{tab:hapnest_benchmark} summarizes the results, where $\overline{\mathrm{FDP}}$ and $\overline{\mathrm{TPP}}$ denote the empirical mean FDP and TPP across Monte Carlo replicates. BH--CA and BY--CA produce zero discoveries across all replicates at both scales: the weak per-SNP effects do not survive the multiple testing correction over $31\,000$ (small-scale) or $394\,000$ (full-scale) marginal tests. BH--HD achieves nonzero power at the small scale but fails to control FDR ($\overline{\mathrm{FDP}} = 11.0\% > \alpha$). BY--HD, which applies the more conservative BY correction to the same screened set, does control FDR ($\overline{\mathrm{FDP}} = 1.0\%$) but at reduced power ($\overline{\mathrm{TPP}} = 28.4\%$). Knockoffs--HD controls FDR but has zero power, as the conservatism of the knockoff filter further compounds the power loss from data splitting. Model-X knockoffs completes only 5 of 100 small-scale replicates within 24\,h and produces zero discoveries in each. At full scale, BH--HD, BY--HD, Knockoffs--HD, and Model-X knockoffs all time out on all 30 replicates, so their full-scale FDR control and power remain unknown. VD--T--Rex is the only method that produces results at both scales, simultaneously controlling FDR ($\overline{\mathrm{FDP}}$ of $5.2\%$ and $5.8\%$, both well below $\alpha = 10\%$) and achieving nonzero power ($\overline{\mathrm{TPP}}$ of $50.6\%$ and $59.1\%$). At full scale, both VD--T--Rex and the marginal baselines read genotype columns on demand from a memory-mapped file rather than loading the full matrix into RAM, achieving peak RSS of approximately $16$\,GB and $4$\,GB, respectively. Despite the additional I/O overhead of memory-mapped access, VD--T--Rex completes $23$ of $30$ full-scale replicates well within the $24$\,h limit. The marginal baselines (BH--CA, BY--CA) complete faster but lack any power at both scales.
\section{Conclusion}
\label{sec:Conclusion}

This paper develops a virtual dummy construction for dummy-augmented forward-selection procedures. The central idea is to avoid explicit design augmentation by generating only the low-dimensional dummy information that forward selection actually uses, namely sequential projections onto the adaptively constructed selection directions. This yields a forward-selection trajectory that is statistically indistinguishable from the explicitly augmented procedure, while substantially reducing memory and dummy-related computational overhead.

Our main theoretical contribution is an exact, finite-sample distributional equivalence between explicitly augmented and virtual-dummy forward selection under rotationally invariant base laws, including Gaussian and spherical dummies. This equivalence holds pathwise. The entire selection trajectory, together with all intermediate quantities used by the algorithm, has the same probability law. Consequently, the virtual construction serves as a drop-in replacement inside procedures such as the T--Rex selector, preserving both its early-stopping rule and its FDR calibration guarantees.

Beyond rotational invariance, we establish a pathwise universality result. For a fixed number of forward-selection steps, the selection path generated by generic standardized i.i.d.\ dummy coordinates converges in distribution to the same limiting law as in the Gaussian case. The proof hinges on a conditional CLT for projections onto data-dependent but delocalized directions, combined with a finite-set stability property ensuring that unselected dummies do not affect the adaptive basis. This provides a universality guarantee for non-rotationally-invariant i.i.d.\ dummy constructions, despite the history-dependent nature of forward selection.

On the practical side, simulations confirm the finite-sample equivalence, illustrate the conditional Gaussianity mechanism underlying universality, and demonstrate substantial reductions in memory usage and runtime relative to explicit dummy augmentation. Experiments on realistic GWAS data with complex linkage disequilibrium further demonstrate that the proposed approach remains computationally feasible and statistically effective at genome-wide scale. We also quantify an extreme-norm inflation effect. Even when Gaussian dummies are variance-matched to spherical dummies, random radial fluctuations can inflate the maximum dummy correlation, making Gaussian dummies overly competitive in early selection steps, leading to a loss in statistical power. This highlights the importance of separating radial and angular effects when interpreting dummy behavior.

Several directions remain for future work. First, extending the universality argument to regimes in which the dummy pool size $L$ grows with $n$ would require uniform control of conditional projection laws over an increasing set of competitors and a careful treatment of extreme order statistics. Second, while our exact equivalence covers rotationally invariant dummy laws, it would be valuable to develop analogous virtual constructions for correlated dummies designed to mimic aspects of the dependence structure of $\bm X$.  Finally, incorporating additional path algorithms and greedy variants, including procedures with variable deletion or more general pursuit rules, would further broaden the scope of virtual dummies as a computational primitive for error-controlled high-dimensional selection.
\section*{Acknowledgments}
{The work of T. Koka, J. Machkour and M. Muma has been funded by the ERC Starting
Grant ScReeningData, Grant No. 101042407. The work of D. P. Palomar was supported by the Hong Kong GRF 16206123 research grant. The authors gratefully acknowledge the computing time provided to them at the NHR Center NHR4CES at TU Darmstadt (project number p0020087). This is funded by the Federal Ministry of Research, Technology and Space, and the state governments participating on the basis of the resolutions of the GWK for national high performance computing at universities (www.nhr-verein.de/unsere-partner).}

\bigskip

\bibliographystyle{plainnat}
\bibliography{references}  

\newpage

\appendix
\section{Compatibility of Variable Selectors}
\label{app:compatible}
Compatibility with the virtual-dummy construction is determined not by a specific forward-selection algorithm, but by whether the quantities used in the selection rule are contained in the revealed subspace. In particular, the framework is not tied to a particular construction of $\V_k$. What matters is only that $\V_k$ be $\F_k$-measurable and contain the directions needed to compute the dummy scores at step $k$. In the linear-model setting, $\V_k$ is conveniently constructed by orthogonalizing each newly selected predictor against the current basis. Equivalently, one may construct the same revealed subspace by sequentially orthogonalizing the residual direction. Thus, validity does not depend on whether $\V_k$ is generated from selected predictors, but only on whether the quantities used to score the dummies are determined by $\F_k$-measurable directions contained in $\V_k$.

This viewpoint extends the framework beyond linear models. In particular, any forward-selection procedure whose selection rule at step $k$ depends on inner products with one or more $\F_k$-measurable score directions is compatible with the framework, provided these directions are included in $\V_k$. This includes score-based GLM forward selection, where the score is the current GLM score vector, as well as robust forward-selection procedures such as LAD- or Huber-based methods, where the score is a transformed residual direction.

By contrast, procedures are not compatible in general when the selection rule cannot be reduced to inner products with $\F_k$-measurable directions contained in $\V_k$. This includes, for example, exact deviance-based GLM selection and exact Wald or likelihood-ratio procedures, where the next selection requires variable-specific refits or curvature information. Nonlinear screening and projection-pursuit methods, whose decision criteria depend on nonlinear functionals not recoverable from projected dummies alone, likewise fall outside the scope of the framework.

\begin{table}[h]
\centering
\caption{Compatibility of common variable--selection procedures with the virtual--dummy framework. A procedure is compatible if its selection rule at step $k$ can be computed from inner products with $\F_k$-measurable directions contained in the revealed subspace $\V_k$. Procedures requiring candidate-specific refits, curvature information, or nonlinear functionals are not compatible in general.}
\label{tab:scope_compatibility}
\small
\renewcommand{\arraystretch}{1.35}
\setlength{\tabcolsep}{6pt}
\begin{tabularx}{\textwidth}{p{5.1cm}p{2.1cm}X}
\toprule
\textbf{Procedure} & \textbf{Compatible?} & \textbf{Reason} \\
\midrule

OMP and greedy forward stepwise \citep{Mallat1993,pati1993orthogonal,Hocking1976,Chen1989,Joseph2013} &
Yes &
The selection rule is determined by correlations with the residual $\bm r_k\in\V_k$, i.e., by dummy scores $\langle \bm d_\ell,\bm r_k\rangle$. \\[0.3em]

Score-based (adaptive) forward stepwise \citep{Tibshirani2015,Zhang2026} &
Yes &
The selection rule uses a score direction $\bm s_k$ given by a transformed residual, e.g.\ $\bm s_k=\bm y-g(\bm X_{\A_k}\bm\beta_k)$ in GLMs; dummy scores are $\langle \bm d_\ell,\bm s_k\rangle$ once $\bm s_k\in\V_k$. \\[0.3em]

Piecewise-linear solution path methods (e.g., LARS/Lasso) \citep{Efron2004,Rosset2007,Zou2005,Zhao2009} &
Yes &
Path events are driven by score/correlation quantities involving $\F_k$-measurable directions $\bm s_k$; the required dummy scores are of the form $\langle \bm d_\ell,\bm s_k\rangle$. \\[0.3em]

LAD / Huber / robust forward selection \citep{Blanchet2019,Chichua2025} &
Yes &
The score direction is a transformed residual, e.g.\ $\bm s_k=\psi(\bm r_k)$; compatibility holds once $\bm s_k\in\V_k$. \\

\midrule

Exact deviance/Wald/LR GLM selection, nonlinear screening, and projection pursuit \citep{McCullagh1989,Sakate2012,Fan2008,Friedman1981,Hwang2020} &
Not generally &
Require candidate-specific refits, curvature information, or nonlinear functionals not recoverable from projected dummies alone. \\

\bottomrule
\end{tabularx}
\end{table}

A practical caveat arises for procedures that trace a regularized solution path. In piecewise-linear path methods, the relevant score direction may change not only when a new predictor enters the model, but also at intermediate path events, such as coefficients hitting zero or loss-function knots being crossed. Consequently, $\V_k$ may need to be enlarged even when no new variable is selected, so its dimension can grow faster than the active set. Moreover, when a variable first enters along such a path, its coefficient is still zero, so the score has not yet moved in the corresponding new direction; to obtain a nonzero component along that direction, one must first take a positive step along the next linear segment of the path before orthogonalizing the updated score.

This issue does not arise in the same way for forward stepwise procedures, where the score is recomputed only after a new predictor is selected. In that case, the dimension of $\V_k$ grows essentially in step with the active set, including for GLM and robust forward-selection procedures. Even for path-following methods, however, the framework can still be beneficial when the solution path is terminated early. Table~\ref{tab:scope_compatibility} summarizes these distinctions.

\section{Proofs}
\label{app:B}
This appendix collects the proofs of all results stated in Sections~\ref{sec:stick_breaking}--\ref{sec:Var_sel_FDR}. To avoid repetition, we do not restate the lemmas and theorems here; all statements appear in the main text together with surrounding intuition and motivation. The arguments below use the same notation and numbering as in the main text.

The structure of this section mirrors the logical development of the paper: Lemma~\ref{lem:rot-invariance} restates the well-known rotational invariance of the Gaussian base laws; Lemma~\ref{lem:conditional-uniformity-affine} proves conditional uniformity on adaptive subspheres in $H$; Lemma~\ref{lem:adaptive_stick_breaking} derives the sequential Dirichlet structure needed for the virtual--dummy sampler; Corollary~\ref{cor:alg_dist} validates the adaptive stick-breaking algorithm; Theorem~\ref{thm:vd-fs-dist} shows distributional equivalence of the augmented and virtual forward selectors; Corollary~\ref{cor:fdr-rot-invariant} transfers the FDR control guarantee to the virtual--dummy T--Rex procedure.

The final part of the appendix establishes the universality results: Lemma~\ref{lem:conditional-clt} proves a conditional CLT for projections onto data--dependent delocalized directions; Lemma~\ref{lem:finite-dummy} upgrades this to joint Gaussian limits for any fixed finite subset of dummy variables; Lemma~\ref{lem:fresh-step} extends the argument to the fresh projections revealed at an arbitrary forward--selection step; and Theorem~\ref{thm:pathwise-univ} concludes by showing that the entire finite selection path has the same asymptotic law under non-Gaussian i.i.d.\ dummy coordinates as under Gaussian dummies. We proceed to the proofs.

\begin{proof}[\textbf{Proof of Lemma~\ref{lem:rot-invariance}}]
Let $\bm d=\operatorname{vec}(\bm D^\top)
=(\bm d_1^\top,\dots,\bm d_n^\top)^\top$.
Independence of the rows of $\bm D$, each with $\mathrm{Cov}(\bm d_i)=\bm \Sigma$, gives $\operatorname{Cov}(\bm D)=\bm I_n\otimes\bm\Sigma$. Since $\operatorname{vec}((\bm Q\bm D)^\top)=(\bm Q\otimes\bm I_L)\bm d$ and
$\bm Q$ is orthogonal,
\[
  (\bm Q\otimes\bm I_L)\bm d
  \sim \mathcal N_{nL}\!\bigl(\bm 0,(\bm Q\bm Q^{\!\top})\!\otimes\!\bm\Sigma\bigr)
  = \mathcal N_{nL}(\bm 0,\bm I_n\otimes\bm\Sigma),
\]
so $\bm Q\bm D\stackrel d=\bm D$.
\end{proof}

\begin{proof}[\textbf{Proof of Lemma~\ref{lem:conditional-uniformity-affine}}]
Let $\{\bm f_1, \dots ,\bm f_m\}$ denote any fixed orthonormal basis for $H$. By rotational invariance of \(S_H\),
the joint distribution of coordinates
\((\langle\bm d,\bm f_1\rangle,\dots,\langle\bm d,\bm f_m\rangle)\)
is unaffected by any orthogonal change of basis in \(H\).
Fix $k$ and $\ell$ and work on the event $\{\tau_\ell>k\}$, so dummy $\ell$ is not yet realized by time $k$ and only its first $k$ projections are revealed in $\F_k^+$. Condition on \( \F_k^+ \), i.e., fix the realized
\(\bm e_1,\dots,\bm e_k\) and coefficients \(\upalpha_1,\dots,\upalpha_k\).
The feasible set for \(\bm d_\ell\) given this information is
the intersection \(S_H\cap\As_k\),
which is a sphere of radius \(R_k=\sqrt{1-\sum_{i=1}^k\upalpha_i^2}\)
embedded in the affine subspace
\(\As_k = \bm d_\ell^{\parallel} +\V_k^\perp\),
where \(\bm d_\ell^{\parallel}=\sum_{i=1}^k \upalpha_i\bm e_i\).

The subgroup
\[
\mathrm O(\V_k^\perp \mid \V_k)
  := \{\bm Q\in\mathrm{O}(H) : \bm Q\bm v=\bm v \text{ for all } \bm v\in\V_k\}
\]
acts transitively on \(S_H\cap\As_k\)
and leaves it invariant.
Because the law of \(\bm d_\ell\) is induced by Haar measure on \(\mathrm{O}(H)\),
the conditional law of \(\bm d_\ell\mid\F_k^+\)
must be invariant under this subgroup.
There exists a unique invariant probability measure under this action,
the uniform (Haar) measure on the subsphere.
Hence
\[
\bm d_\ell \mid \F_k^+
   \sim \Unif(S_H\cap\As_k),
\]
which implies that the residual component
\(\bm d_\ell^\perp = \bm d_\ell-\bm d_\ell^{\parallel}\)
has direction uniformly distributed on $S_H\cap\V_k^\perp$ and radius $R_k = \sqrt{1 - \sum_{i=1}^k \upalpha_i^2}$, as stated.
\end{proof}

\begin{proof}[\textbf{Proof of Lemma~\ref{lem:adaptive_stick_breaking}}]
After revealing $\upalpha_{1\ell},\dots,\upalpha_{k\ell}$, recall that $Z_i:=\upalpha_{i\ell}^2$, and set
\[
  R_{\ell,k}^2 := 1-\sum_{i=1}^k Z_i = \|\bm d_\ell^\perp\|^2 .
\]
Let
\(\widetilde{\bm d}_\ell^\perp := \bm d_\ell^\perp / \|\bm d_\ell^\perp\|
\in S_H\cap\V_k^\perp\).
By Lemma~\ref{lem:conditional-uniformity-affine},
\[
\widetilde{\bm d}_\ell^\perp\mid\F_k^+
~\sim~
\Unif\bigl(S_H\cap\V_k^\perp\bigr).
\]
For any \(\F_k^+\)-measurable
\(\bm e_{k+1}\in\V_k^\perp\),
\[
Z_{k+1}
= R_{\ell,k}^2\,\langle \widetilde{\bm d}_\ell^\perp, \bm e_{k+1}\rangle^2,
\qquad
\langle \widetilde{\bm d}_\ell^\perp, \bm e_{k+1}\rangle^2
~\big|~\F_k^+
\sim \mathrm{Beta}\Bigl(\tfrac12,\tfrac{m-k-1}{2}\Bigr)
=: U_1.
\]
Hence \(Z_{k+1}=R_{\ell,k}^2 U_1\).
Iterating with $U_i \mid \F_{k+i-1}^+ \sim \mathrm{Beta}(\tfrac12,\tfrac{m-k-i}{2})$ and using the standard stick-breaking characterization of the Dirichlet distribution yields the claimed Dirichlet law for $(Z_1,\dots,Z_m)$ and its conditional marginals.
\end{proof}

\begin{proof}[\textbf{Proof of Corollary~\ref{cor:alg_dist}}]
By Lemma~\ref{lem:adaptive_stick_breaking}, the squared coefficients 
\((\upalpha_1^2,\dots,\upalpha_m^2)\) follow the 
\(\mathrm{Dirichlet}(\tfrac12,\dots,\tfrac12)\) law in \(H\), even when the basis 
\((\bm e_1,\dots,\bm e_m)\) is chosen adaptively.
At the final step, the revealed subspace is one-dimensional in \(H\), so the last coefficient 
absorbs the remainder deterministically:
\(\upalpha_m = S_m R_{m-1}\).
This is precisely the procedure in 
Algorithm~\ref{alg:seq_spherical}, hence 
\(\bm d = \sum_{k=1}^m \upalpha_k\bm e_k\)
is uniformly distributed on \(S_H\).
\end{proof}

\begin{proof}[\textbf{Proof of Theorem~\ref{thm:vd-fs-dist}}]
We show that $\mathcal{P}_{\mathrm{ad}}$ and $\mathcal{P}_{\mathrm{vd}}$
share the same conditional transition kernel from $\F_k$ to $\F_{k+1}$
for every $k$; since both start from $\sigma(\bm X,\bm y)$ and terminate
in finitely many steps, the chain rule of probability then gives
$\mathcal{P}_{\mathrm{ad}} \stackrel{d}{=} \mathcal{P}_{\mathrm{vd}}$.

At step $k$, passing from $\F_k$ to $\F_k^+$ realizes at most one dummy
and fixes $\bm e_{k+1}$. Since $\bm D$ has product law
$\Psi(H)^{\otimes L}$, this does not reveal information about the
orthogonal residuals $\{\bm d_{\ell,k}^{\perp} : \tau_\ell>k\}$, which remain conditionally i.i.d.\ given $\F_k^+$. The next projections are $\upalpha_{k+1,\ell} = \langle\bm d_{\ell,k}^{\perp},\bm e_{k+1}\rangle$, and by rotational invariance:

\noindent
(\,G\,) under the Gaussian base law, isotropy of $\bm d_{\ell,k}^{\perp}$ in $\V_k^\perp$ gives $\upalpha_{k+1,\ell}\mid\F_k^+\sim\mathcal N(0,1)$, conditionally i.i.d.

\noindent
(\,S\,) under the spherical base law, Lemma~\ref{lem:conditional-uniformity-affine} gives uniformity on $S_H\cap\V_k^\perp$ with radius $R_{\ell,k}=\sqrt{1-\sum_{i=1}^k\upalpha_{i\ell}^2}$, hence $\upalpha_{k+1,\ell}\mid\F_k^+ \sim R_{\ell,k}\sqrt{\mathrm{Beta}(\tfrac12,\tfrac{m-k-1}{2})}$, again conditionally i.i.d.

\noindent
By construction, VD--FS samples the projections
$\{\upalpha_{k+1,\ell}\}_{\tau_\ell>k}$ from exactly these conditional laws. All subsequent updates are deterministic functions of $\F_k^+$-measurable quantities and the fresh projections $\{\upalpha_{k+1,\ell}\}_{\tau_\ell>k}$, so the two processes share the same transition kernel at every step.
\end{proof}

\begin{proof}[\textbf{Proof of Corollary~\ref{cor:fdr-rot-invariant}}]
Fix $(\bm X,\bm y)$. For each $j\in\{1,\dots,p\}$, the relative occurrence is
\[
\Phi_{T,L}(j)
= \frac{1}{B}\sum_{b=1}^B \mathds{1}\!\left\{\, j \in \mathcal{C}_{b,L}(T) \,\right\}.
\]
By Theorem~\ref{thm:vd-fs-dist}, for each random experiment $b\in\{1,\dots,B\}$, the virtual-dummy and explicitly augmented procedures induce the same conditional law on the full forward-selection trajectory given
$(\bm X,\bm y)$. Hence the indicators
\[
\mathds{1}\!\left\{\, j \in \mathcal{C}_{b,L}(T) \,\right\},
\qquad b=1,\dots,B,
\]
have the same conditional law under $\mathcal P_{\mathrm{aug}}$ and $\mathcal P_{\mathrm{vd}}$. Since the $B$ random experiments are generated independently, these indicators are conditionally i.i.d.\ given $(\bm X,\bm y)$ under both procedures. Therefore,
\[
\Phi_{T,L}(j)
\xrightarrow[B\to\infty]{a.s.}
\Phi_{T,L}^{\infty}(j)
:=
\E\!\left[
\mathds{1}\!\left\{\, j \in \mathcal{C}_{1,L}(T) \,\right\}
\,\middle|\, \bm X,\bm y
\right],
\]
and the deterministic conditional limit $\Phi_{T,L}^{\infty}(j)$ is the same under virtual dummies and explicit augmentation.

It follows that all asymptotic quantities used by the T--Rex selector are the same under both constructions, since they are deterministic functionals of $\{\Phi_{T,L}^{\infty}(j)\}_{j=1}^p$. In particular, the limiting selected set
\[
\widehat{\mathcal A}_L(v,T)=\{j:\Phi_{T,L}^{\infty}(j)>v\},
\]
the false discovery proportion $\mathrm{FDP}(v,T,L)$, the corresponding estimator $\widehat{\mathrm{FDP}}(v,T,L)$, and the induced voting threshold $v$ coincide conditional on $(\bm X,\bm y)$. Therefore, the proof of Theorem~1 in \citet{Machkour2025} applies without further modifications, yielding
\[
\E[\mathrm{FDP}(v,T,L)] \le \alpha.
\]
\end{proof}

\begin{lem}
\label{lem:conditional-clt}
For each $n$, let $\bm\updelta^{(n)} = (\updelta^{(n)}_1,\dots,\updelta^{(n)}_n)^\top$ have
i.i.d.\ coordinates with $\E[\updelta]=0$ and $\E[\updelta^2]=1$.
Let $\G_n$ be a $\sigma$--field such that $\bm\updelta^{(n)}$ is independent of
$\G_n$.
Suppose there exist $\G_n$--measurable unit vectors
$\bm e_1^{(n)},\dots,\bm e_{k+1}^{(n)}\in\R^n$ forming an orthonormal system and satisfying
\[
  \max_{1\le i\le k+1}\|\bm e_i^{(n)}\|_\infty \xrightarrow[n\to\infty]{P} 0.
\]
Let $\bm E_n = (\bm e_1^{(n)},\dots,\bm e_{k+1}^{(n)})\in\R^{n\times(k+1)}$ and
define the projection vector
\(
  \bm\upalpha_n := \bm E_n^\top \bm\updelta^{(n)} \in \R^{k+1}.
\)
Then
\[
  \bm\upalpha_n \,\big|\, \G_n
  \xRightarrow[n\to\infty]{d}\;
  \mathcal N(\bm 0,\bm I_{k+1})
  \quad\text{in probability}.
\]
\end{lem}

\begin{proof}
Fix $\bm t\in\R^{k+1}$ with $\|\bm t\|_2=1$ and set
\[
  \xi_n := \bm t^\top \bm\upalpha_n
         = \sum_{j=1}^n q_j^{(n)} \updelta_j^{(n)},
  \qquad
  \bm q^{(n)} := \bm E_n \bm t
  \in \R^n.
\]
Conditional on $\G_n$, the weights $\bm q^{(n)}$ are deterministic and
$\{\updelta_j^{(n)}\}_{j=1}^n$ are i.i.d.\ with mean $0$ and variance $1$.
By orthonormality of $\{\bm e_i^{(n)}\}$,
\[
  \|\bm q^{(n)}\|_2^2
  = \Big\|\sum_{i=1}^{k+1} t_i \bm e_i^{(n)}\Big\|_2^2
  = \sum_{i=1}^{k+1} t_i^2
  = 1,
\]
so $\Var(\xi_n\mid\G_n)=1$. By Cauchy--Schwarz and the delocalization assumption,
\[
  \max_{1\le j\le n} |q_j^{(n)}|
  \le \sum_{i=1}^{k+1} |t_i|\,\|\bm e_i^{(n)}\|_\infty
  \xrightarrow[n\to\infty]{P} 0.
\]
Thus, the \emph{Noether condition}
\(\{
  \max_{1\le j\le n} (q_j^{(n)})^2 \;\longrightarrow\; 0\}
\)
holds in probability (since $\|\bm q^{(n)}\|_2^2=1$). Let $(n_r)$ be any subsequence. There exists a further subsequence
$(n_{r_s})$ along which
\[
  \max_{1\le j\le n_{r_s}} |q_j^{(n_{r_s})}|
  \xrightarrow[s\to\infty]{a.s.} 0.
\]
Fix $\omega$ in this probability-one event.
For this $\omega$, the weights $\bm q^{(n_{r_s})}(\omega)$ are deterministic,
satisfy $\sum_j (q_j^{(n_{r_s})}(\omega))^2=1$ and
$\max_j |q_j^{(n_{r_s})}(\omega)|\to0$.  
By the Hájek--Šidák central limit theorem
\citep[Thm.~1, Sec.~6.1.2]{Hjek1999},
the weighted sums
\[
  \xi_{n_{r_s}}(\omega)
  = \sum_{j=1}^{n_{r_s}} q_j^{(n_{r_s})}(\omega)\,\updelta_j^{(n_{r_s})}
\]
converge in distribution to $\mathcal N(0,1)$.
Equivalently, for every bounded continuous $f:\R\to\R$,
\[
  \E\bigl[f(\xi_{n_{r_s}})\mid\G_{n_{r_s}}\bigr](\omega)
  \longrightarrow \E[f(Z)],
  \qquad Z\sim\mathcal N(0,1).
\]

Thus along $(n_{r_s})$ we have
$\E[f(\xi_{n_{r_s}})\mid\G_{n_{r_s}}]\to\E[f(Z)]$ almost surely, hence in
probability.  Since every subsequence $(n_r)$ admits such a further subsequence,
it follows that
\[
  \xi_n \mid \G_n \xRightarrow[n\to\infty]{d} \mathcal N(0,1)
  \quad\text{in probability}.
\]

Because this holds for every $\bm t\in\R^{k+1}$, the (conditional) Cramér--Wold
device implies
\[
  \bm\upalpha_n \mid \G_n
  \xRightarrow[n\to\infty]{d} \mathcal N(\bm 0,\bm I_{k+1})
  \quad\text{in probability},
\]
which is the desired conclusion.
\end{proof}

\begin{lem}
\label{lem:finite-dummy}
Work under Assumption~\ref{ass:universality}.  
Fix a finite index set $\Lambda\subset\{1,\dots,L\}$ and define
\(
\G_{n,\Lambda}
  := \sigma\!\bigl(\bm X,\bm y,\{\bm d_r^{(n)}:r\notin\Lambda\}\bigr).
\)
For $\ell\in\Lambda$ and $i=1,\dots,K$ define the scaled projections
\[
  \upalpha_{i,\ell}^{(n)}
  := \langle \bm d_\ell^{(n)},\bm e_i^{(n)}\rangle,
  \qquad
  \bm\upalpha_\ell^{(n)}
  := (\upalpha_{1,\ell}^{(n)},\dots,\upalpha_{K,\ell}^{(n)})^\top.
\]

On the event 
\(
  \Lambda\cap\{J_1^{(n)},\dots,J_{K-1}^{(n)}\}=\emptyset,
\)
i.e., no index in $\Lambda$ is selected before step $K$, we have
\[
  \bigl(\bm\upalpha_\ell^{(n)}\bigr)_{\ell\in\Lambda}
  \,\Big|\,\G_{n,\Lambda}
  \xRightarrow[n\to\infty]{d}
  \bigl(\bm Z_\ell\bigr)_{\ell\in\Lambda}
  \quad\text{in probability},
\]
where $\{\bm Z_\ell\}_{\ell\in\Lambda}$ are independent 
$\mathcal N(\bm 0,\bm I_K)$ random vectors.
\end{lem}

\begin{proof}
By \textbf{(Sel)}, on the event that no $\ell\in\Lambda$ is selected in the first 
$K-1$ steps, the forward selector that operates on the leave--$\Lambda$--out design  (i.e.\ using only the dummies $\{\bm d_r^{(n)}: r\notin\Lambda\}$) produces  the same sequence of directions $\bm e_1^{(n)},\dots,\bm e_K^{(n)}$. Thus, on this event, each $\bm e_i^{(n)}$ is $\G_{n,\Lambda}$--measurable.

Fix $\ell\in\Lambda$. The raw dummy $\bm\updelta_\ell^{(n)}$ is independent  of $\G_{n,\Lambda}$ and has i.i.d.\ coordinates with mean $0$ and variance  $1$. Applying Lemma~\ref{lem:conditional-clt} with $k+1$ replaced by $K$ and with $\G_n=\G_{n,\Lambda}$ yields
\[
\bm\vartheta_\ell^{(n)}
  := \bigl(\langle\bm\updelta_\ell^{(n)},\bm e_1^{(n)}\rangle,\dots,
           \langle\bm\updelta_\ell^{(n)},\bm e_K^{(n)}\rangle\bigr),
\qquad
\bm\vartheta_\ell^{(n)} \,\Big|\,\G_{n,\Lambda}
  \xRightarrow[n\to\infty]{d}
  \mathcal N(\bm 0,\bm I_K)
  \quad\text{in probability}.
\]
Since the $\{\bm\updelta_\ell^{(n)}\}_{\ell\in\Lambda}$ are independent, we obtain joint convergence over $\ell\in\Lambda$:
\[
  \bigl(\bm\vartheta_\ell^{(n)}\bigr)_{\ell\in\Lambda}\,\Big|\,\G_{n,\Lambda}
  \xRightarrow[n\to\infty]{d}
  \bigl(\bm Z_\ell\bigr)_{\ell\in\Lambda}
  \quad\text{in probability},
\]
where the $\bm Z_\ell$ are independent $\mathcal N(\bm 0,\bm I_K)$.

Finally, write
\[
  \upalpha_{i,\ell}^{(n)}
  = \langle\bm d_\ell^{(n)},\bm e_i^{(n)}\rangle
  = \frac{\vartheta_{i,\ell}^{(n)}}{\rho_\ell^{(n)}},
\qquad
\rho_\ell^{(n)} := \|I_H\bm\updelta_\ell^{(n)}\|_2/\sqrt n.
\]
By the law of large numbers, $\rho_\ell^{(n)}\to1$ in probability, and  $\rho_\ell^{(n)}$ is independent of $\G_{n,\Lambda}$.  Conditional Slutsky (componentwise and jointly over $\ell\in\Lambda$) yields
\[
  \bigl(\bm\upalpha_\ell^{(n)}\bigr)_{\ell\in\Lambda}\,\Big|\,\G_{n,\Lambda}
  \xRightarrow[n\to\infty]{d}
  \bigl(\bm Z_\ell\bigr)_{\ell\in\Lambda}
  \quad\text{in probability},
\]
which is the claim.
\end{proof}

\begin{lem}
\label{lem:fresh-step}
Work under Assumption~\ref{ass:universality}. Fix $k\in\{0,\dots,K-1\}$. Let
\(
  \Jset_k^{(n)} := \{\ell:\tau_\ell^{(n)}>k\}
\)
be the (random) set of dummies not yet selected at step $k$. For each $\ell\in \Jset_k^{(n)}$ define
\(
  \upalpha_{k+1,\ell}^{(n)}
  := \langle \bm d_\ell^{(n)},\bm e_{k+1}^{(n)}\rangle.
\)
Then, conditional on $\F_k^+$,
\[
  \bigl(\upalpha_{k+1,\ell}^{(n)}\bigr)_{\ell\in\Jset_k^{(n)}}
  \,\Big|\,\F_k^+
  \xRightarrow[n\to\infty]{d}
  \bigl(Z_{k+1,\ell}\bigr)_{\ell\in\Jset_k^{(n)}}
  \quad\text{in probability},
\]
where $\{Z_{k+1,\ell}\}_{\ell\in\Jset_k^{(n)}}$ are i.i.d.\ $\mathcal N(0,1)$.
\end{lem}

\begin{proof}
Fix a finite deterministic subset $\Lambda\subset\{1,\dots,L\}$ and define
\[
\G_{n,\Lambda}
  := \sigma\!\bigl(\bm X,\bm y,\{\bm d_r^{(n)}:r\notin\Lambda\}\bigr).
\]
For each $n$, consider the event
\(
  E_n := \{ \Lambda \cap \Jset_k^{(n)} = \Lambda \},
\)
on which no index in $\Lambda$ has been selected before step $k$. On $E_n$, Lemma~\ref{lem:finite-dummy} (with $K$ replaced by $k+1$) yields
\[
  \bigl(\upalpha_{i,\ell}^{(n)} : 1\le i\le k+1,\ \ell\in\Lambda\bigr)
  \,\Big|\, \G_{n,\Lambda}
  \xRightarrow[n\to\infty]{d}
  \bigl(Z_{i,\ell} : 1\le i\le k+1,\ \ell\in\Lambda\bigr)
  \quad\text{in probability},
\]
where the $Z_{i,\ell}$ are i.i.d.\ $\mathcal N(0,1)$. Define
\[
  \bm U_{n,\Lambda}
  := \left(\upalpha_{i,\ell}^{(n)}\right)_{1\le i\le k,\ \ell\in\Lambda},
  \qquad
  \bm V_{n,\Lambda}
  := \left(\upalpha_{k+1,\ell}^{(n)}\right)_{\ell\in\Lambda},
\]
and let $\bm Z_{\mathrm{past}}$ and $\bm Z_{\mathrm{fut}}$ denote the corresponding Gaussian limits.
Because the coordinates $Z_{i,\ell}$ are independent, the random vectors
$\bm Z_{\mathrm{past}}$ and $\bm Z_{\mathrm{fut}}$ are independent.

On $E_n$, the $\sigma$--field $\F_k^+$ can be generated as
\[
  \F_k^+
  = \sigma\bigl(\G_{n,\Lambda},\,\bm U_{n,\Lambda},\,\bm W_n\bigr),
\]
where $\bm W_n$ collects all additional information in $\F_k^+$ (i.e.,  projections and realized dummies outside $\Lambda$, and tie-breaking randomness). 
Conditional on $\G_{n,\Lambda}$ and $\bm U_{n,\Lambda}$, the fresh projections
$\bm V_{n,\Lambda}$ depend only on the dummies
$\{\bm d_\ell^{(n)}:\ell\in\Lambda\}$, which are independent of $\bm W_n$.
Hence, for any bounded measurable $f$,
\[
  \E\bigl[f(\bm V_{n,\Lambda}) \,\big|\, \G_{n,\Lambda},\bm U_{n,\Lambda},\bm W_n\bigr]
  = \E\bigl[f(\bm V_{n,\Lambda}) \,\big|\, \G_{n,\Lambda},\bm U_{n,\Lambda}\bigr].
\]

Now fix a bounded continuous $f:\R^{|\Lambda|}\to\R$ and set
\[
  \varpi_n := \E\bigl[f(\bm V_{n,\Lambda}) \,\big|\, \G_{n,\Lambda},\bm U_{n,\Lambda}\bigr].
\]
By the joint conditional convergence of $(\bm U_{n,\Lambda},\bm V_{n,\Lambda})$ 
to $(\bm Z_{\mathrm{past}},\bm Z_{\mathrm{fut}})$ given $\G_{n,\Lambda}$, and the
independence of these limit components, a standard diagonal subsequence argument
(as in the proof of Lemma~\ref{lem:conditional-clt}) yields
\[
  \varpi_n \xrightarrow[n\to\infty]{P} \E\bigl[f(\bm Z_{\mathrm{fut}})\bigr].
\]

Combining the above,
\[
  \E\bigl[f(\bm V_{n,\Lambda}) \,\big|\, \F_k^+\bigr]
  \xrightarrow[n\to\infty]{P} \E\bigl[f(\bm Z_{\mathrm{fut}})\bigr]
\]
on the event $E_n$.

Because $L$ is fixed, the random set $\Jset_k^{(n)}$ takes values in a finite collection of subsets of $\{1,\dots,L\}$. For any fixed subset $S\subset\{1,\dots,L\}$, we may apply the above argument with $\Lambda=S$ on the event $\{\Jset_k^{(n)}=S\}$. Summing over all such $S$ then yields the claim for $\bigl(\upalpha_{k+1,\ell}^{(n)}\bigr)_{\ell\in\Jset_k^{(n)}}$.
\end{proof}

\begin{proof}[\textbf{Proof of Theorem~\ref{thm:pathwise-univ}}]
For each $k\in\{0,\dots,K-1\}$, the $(k{+}1)$-st selection can be written as
\[
  J_{k+1}^{(n)}
  =
  \phi_{k+1}^{(n)}\!\Bigl(
    F_k^{(n)},\,
    \bigl(\upalpha_{k+1,\ell}^{(n)}\bigr)_{\ell\in\Jset_k^{(n)}}
  \Bigr),
\]
where $F_k^{(n)}$ is a random element generating the revealed information,
i.e., $\sigma(F_k^{(n)})=\F_k^+$, and $\phi_{k+1}^{(n)}$ is a (Borel) measurable
decision rule encoding the forward--selection update (including deterministic
tie--breaking). Both the non-Gaussian and Gaussian dummy ensembles use the
same measurable mapping $\phi_{k+1}^{(n)}$, applied to their respective fresh
projection vectors. We prove the claim by induction on $k=1,\dots,K$.

\medskip\noindent
\textbf{Induction hypothesis $(\mathrm{IH}(k))$.}
There exists a random vector $(J_1,\dots,J_k)$ such that
\[
  (J_{\mu,1}^{(n)},\dots,J_{\mu,k}^{(n)})
  \xRightarrow[n\to\infty]{d}
  (J_1,\dots,J_k)
\quad\text{and}\quad
  (J_{\mathrm G,1}^{(n)},\dots,J_{\mathrm G,k}^{(n)})
  \xRightarrow[n\to\infty]{d}
  (J_1,\dots,J_k).
\]

\medskip\noindent
\textbf{Base case $k=1$.}
At the first step, the direction $\bm e_1^{(n)}$ depends only on $(\bm X,\bm y)$
and is common to both ensembles. Lemma~\ref{lem:fresh-step} with $k=0$ implies
that, conditional on $\F_0^+=\sigma(\bm X,\bm y)$, the dummy projections
$\{\upalpha_{1,\ell}^{(n)}\}_{\ell=1}^L$ converge in law (in probability) to an
i.i.d.\ $\mathcal N(0,1)$ family, which coincides with the exact conditional law
in the Gaussian ensemble. By Assumption~\ref{ass:universality} \textbf{(GP)}, ties occur with probability
zero in the first $K$ steps, so the maximizer defining $J_1^{(n)}$ is almost
surely unique under the limiting Gaussian law. Hence the mapping
$\phi_1^{(n)}$ is almost surely continuous at the realized projection vector.
The conditional continuous mapping theorem yields
\[
  J_{\mu,1}^{(n)} \xRightarrow[n\to\infty]{d} J_1,
  \qquad
  J_{\mathrm G,1}^{(n)} \xRightarrow[n\to\infty]{d} J_1,
\]
for some random index $J_1$. This establishes $(\mathrm{IH}(1))$.

\medskip\noindent
\textbf{Induction step.}
Assume $(\mathrm{IH}(k))$ holds for some $k\in\{1,\dots,K-1\}$.
Consider step $k+1$.
Let
\[
  \Jset_k^{(n)} := \{\ell:\tau_\ell^{(n)}>k\}
\]
be the unselected dummy indices at step $k$, and let
$(\upalpha_{k+1,\ell}^{(n)})_{\ell\in\Jset_k^{(n)}}$ be the fresh dummy projections. By Lemma~\ref{lem:fresh-step}, conditional on $\F_k^+$, the fresh projection
vector in the non-Gaussian ensemble converges in distribution, in probability,
to an i.i.d.\ $\mathcal N(0,1)$ family. This is exactly the conditional law of
the fresh dummy projections in the Gaussian ensemble given the same history
$\F_k^+$.

The next selected index is obtained by applying the same decision rule:
\[
  J_{\mu,k+1}^{(n)}
  =
  \phi_{k+1}^{(n)}\!\Bigl(
    F_k^{(n)},\,
    (\upalpha_{k+1,\ell}^{(n)})_{\ell\in\Jset_k^{(n)}}
  \Bigr),
\]
and analogously for $J_{\mathrm G,k+1}^{(n)}$. By Assumption~\ref{ass:universality} \textbf{(GP)}, ties occur with probability zero under the Gaussian limit, so $\phi_{k+1}^{(n)}$ is almost surely continuous at the limiting projection vector. Combining (i) the induction hypothesis $(\mathrm{IH}(k))$, (ii) the conditional convergence in Lemma~\ref{lem:fresh-step}, and
(iii) the conditional continuous mapping theorem, we obtain
\[
  (J_{\mu,1}^{(n)},\dots,J_{\mu,k+1}^{(n)})
  \xRightarrow[n\to\infty]{d}
  (J_1,\dots,J_k,J_{k+1}),
\]
and
\[
  (J_{\mathrm G,1}^{(n)},\dots,J_{\mathrm G,k+1}^{(n)})
  \xRightarrow[n\to\infty]{d}
  (J_1,\dots,J_k,J_{k+1}),
\]
for some random index $J_{k+1}$. This is $(\mathrm{IH}(k+1))$. By induction, $(\mathrm{IH}(K))$ holds (with $\mathscr J=(J_1,\dots,J_K)$),
as claimed.
\end{proof}

\end{document}